\documentclass[english]{smfart}
\usepackage{a4wide}

\usepackage[T1]{fontenc}   
\usepackage[latin1]{inputenc}
\usepackage[english]{babel} 
\usepackage{amsmath}
\usepackage{amsfonts}
\usepackage{dsfont}
\usepackage{amssymb}
\usepackage{smfthm}
\usepackage{mathrsfs}
\usepackage{stmaryrd}
\usepackage{url}
\usepackage{color}
\usepackage{enumerate}
\usepackage{hyperref}
\usepackage{multirow}
\usepackage{algorithm}
\usepackage[noend]{algpseudocode}
\usepackage{tikz}
\usepackage{tikz-cd}

\theoremstyle{plain}
\newtheorem{thm}{Theorem}
{\bfseries}{\itshape}
{\bfseries}{\itshape}
\newtheorem{proposition}[thm]{Proposition}
\newtheorem{cor}[thm]{Corollary}
\newtheorem{lem}[thm]{Lemma}
\newtheorem{conjec}[thm]{Conjecture}
\theoremstyle{definition}
\newtheorem{definition}[thm]{Definition}
\newtheorem{example}[thm]{Example}

\theoremstyle{remark}
\newtheorem{remark}[thm]{Remark}

\renewcommand{\leq}{\leqslant} 
\renewcommand{\geq}{\geqslant} 

\newcommand{\F}{\ensuremath{\mathbb{F}}}
\newcommand{\Fq}{\ensuremath{\mathbb{F}_q}}
\newcommand{\Fqm}{\ensuremath{\mathbb{F}_{q^m}}}

\newcommand{\mat}[1]{\ensuremath{\boldsymbol{#1}}}
\newcommand{\code}[1]{\ensuremath{\mathscr{#1}}}

\newcommand{\AC}{\code{A}}

\newcommand{\CC}{\code{C}}
\newcommand{\RC}{\code{R}}

\newcommand{\Bm}{\mat{B}}
\newcommand{\Gm}{\mat{G}}
\newcommand{\Hm}{\mat{H}}
\newcommand{\Mm}{\mat{M}}       
\newcommand{\Pm}{\mat{P}}
\newcommand{\Qm}{\mat{Q}}
\newcommand{\Sm}{\mat{S}}

\newcommand{\av}{\mat{a}}
\newcommand{\bv}{\mat{b}}
\newcommand{\cv}{\mat{c}}

\newcommand{\ev}{\mat{e}}
\newcommand{\mv}{\mat{m}}

\newcommand{\vv}{\mat{v}}
\newcommand{\xv}{{\mat{x}}}
\newcommand{\yv}{{\mat{y}}}

\newcommand{\transpose}[1]{{#1}^{ {\intercal} }}
\newcommand{\transposeb}[1]{{\left(#1\right)}^{ {\intercal} }}
\newcommand{\Tr}{\mathbf{Tr}}
\newcommand{\ExpMat}{\mathbf{ExpMat}}
\newcommand{\ExpCode}{\mathbf{Exp}}
\newcommand{\ExpVect}{\mathbf{Exp}}
\newcommand{\SqueezeVect}{\mathbf{Squeeze}}
\newcommand{\SqueezeMat}{\mathbf{SqueezeMat}}
\newcommand{\SqueezeCode}{\mathbf{Squeeze}}
\newcommand{\Dual}{\mathbf{Dual}}
\newcommand{\Dualb}[1]{\mathbf{Dual}\left( #1 \right)}

\renewcommand{\L}{\mathcal L}

\newcommand{\HmL}{\Hm_{\L}}

\newcommand{\Hsec}{\Hm_{\rm sec}}
\newcommand{\Gsec}{\Gm_{\rm sec}}
\newcommand{\Gpub}{\Gm_{\rm pub}}

\newcommand{\Hpub}{\Hm_{\rm pub}}
\newcommand{\Cexp}{\CC_{\rm exp}}
\newcommand{\Cexptiny}{\CC_{\rm exp,tiny}}
\newcommand{\Ctiny}{\CC_{\rm tiny}}
\newcommand{\Cpub}{\CC_{\rm pub}}

\newcommand{\Grand}{\mat{G}_{\rm rand}}

\newcommand{\RS}[2]{\text{\bf RS}_{#1}(#2)}
\newcommand{\GRS}[3]{\text{\bf GRS}_{#1}(#2,#3)}

\newcommand{\sh}[2]{\mathbf{Sh}_{#2}\left(#1\right)}
\newcommand{\pu}[2]{\mathbf{Pct}_{#2}\left(#1\right)}

\newcommand{\sq}[1]{#1^{\star 2}}
\newcommand{\sqb}[1]{\left(#1\right)^{\star 2}}
\newcommand{\tstar}{\tilde{\star}}
\newcommand{\tsq}[1]{#1^{\tstar 2}}
\newcommand{\tsqb}[1]{\left(#1\right)^{\tstar 2}}

\renewcommand{\SS}{\mathcal{S}}
\newcommand{\CCSS}{\CC_{|\SS}}
\newcommand{\CCSSi}{\CC_{|(\SS_0,\ldots,\SS_{n-1})}}
\newcommand{\SScode}[2]{{#1}_{|#2}}

\newcommand{\B}{\mathcal{B}}
\newcommand{\Bfull}{{\B_{\rm full}}}
\newcommand{\Bdual}{\B^{\ast}}

\newcommand{\Ball}{\B = (b_0, \ldots, b_{m-1})}
\newcommand{\Bdualall}{\Bdual = (b^{\ast}_0, \ldots, b^{\ast}_{m-1})}

\newcommand{\Bgamma}{\B_\gamma}
\newcommand{\Bgammaall}{\Bgamma = (1, \gamma, \ldots, \gamma^{m-1})}
\newcommand{\BS}{{\B_{\SS}}}

\newcommand{\BSp}{{\B_{\SS^2}}}

\newcommand{\IInt}[2]{\llbracket #1, #2 \rrbracket}
\newcommand{\Ical}{\mathcal{I}}
\newcommand{\Lcal}{\mathcal{L}}

\newcommand{\Lexp}{\mathcal{L}'}
\newcommand{\JLM}{\mathcal{J}(\lambda,m)}
\newcommand{\KLM}{\mathcal{K}(\lambda,m)}

\newcommand{\map}[4]{\left\{
\begin{array}{ccc}
#1 & \longrightarrow & #2 \\
#3 & \longmapsto     & #4
        \end{array}
      \right.}

\newcommand{\Sup}{\mathbf{Support}}
\newcommand{\eqdef}{\stackrel{\text{def}}{=}}
\newcommand{\ie}{\textit{i.e.}\,}
\newcommand{\Span}[2]{\left\langle \, #1 \, \right\rangle_{#2}}
\newcommand{\Fqspan}[1]{\left\langle \, #1 \, \right\rangle_{\Fq}}

\newcommand{\Prob}[1]{\mathbb{P} \left[ #1 \right]}
\newcommand{\Esp}[1]{\mathbb{E}\left[ #1 \right]}

\def\longversion{1}

\begin{document}

\author{Alain Couvreur}
\address{Inria}
\address{LIX, \'Ecole polytechnique, Palaiseau, France}
\email{alain.couvreur@inria.fr}

\author{Matthieu Lequesne}
\address{Sorbonne Universit\'e}
\address{Inria, Paris, France}
\email{matthieu.lequesne@inria.fr}

\title[On the security of subspace subcodes of Reed--Solomon
codes]{{On the security of subspace subcodes\\ of Reed--Solomon codes
for public key encryption}}


\maketitle

\begin{abstract}
  This article discusses the security of McEliece-like encryption
  schemes using subspace subcodes of Reed--Solomon codes, \ie subcodes
  of Reed--Solomon codes over $\Fqm$ whose entries lie in a fixed
  collection of $\Fq$--subspaces of $\Fqm$.  These codes appear to be
  a natural generalisation of Goppa and alternant codes and provide a
  broader flexibility in designing code based encryption schemes.  For
  the security analysis, we introduce a new operation on codes called
  the {\em twisted product} which yields a polynomial time
  distinguisher on such subspace subcodes as soon as the chosen
  $\Fq$--subspaces have dimension larger than $m/2$.  From this
  distinguisher, we build an efficient attack which in particular
  breaks some parameters of a recent proposal due to Khathuria,
  Rosenthal and Weger.
\end{abstract}

\tableofcontents


\section{Introduction}

In the late 70's, at the very beginning of public key cryptography,
McEliece proposed a public key encryption scheme whose security relies
on the hardness of the bounded decoding problem \cite{BMT78}. However,
the system should be instantiated with a public code equipped with an
efficient decoding algorithm and one usually expect from the public
code to be indistinguishable from an arbitrary one since this last
property guarantees the security of the system.  In his seminal article
\cite{M78}, McEliece proposed to instantiate his system with a binary
Goppa code~\cite{G70,G71} (see~\cite{B73} for a description in
English).

\medskip

One of the major drawbacks of such a proposal is the significant size
of the public key. McEliece's historical proposal with binary Goppa
codes required a 32.7 kB key for a claimed classical security of 65
bits.  In the recent NIST submission {\em Classic
  McEliece}~\cite{BCLMNPPSSSW19}, a public key size of 261 kB is
proposed for a claimed security level of 128 bits.  For this reason,
there has been many attempts in the last forty years to replace the
Goppa codes used in McEliece's scheme by other families of codes, in
order to reduce the key size. Many of these proposals focus on
Generalised Reed-Solomon (GRS) codes, since they benefit from
excellent decoding properties.  On the other hand, their structure is
difficult to hide.  Consequently, to our knowledge, all code-based
cryptographic schemes using GRS codes or low-codimensional subcodes of
GRS codes as trapdoor have been attacked.

\medskip

The use of GRS codes to replace Goppa codes in McEliece's scheme was
initially suggested in Niederreiter's paper \cite{N86}, but this
proposal was attacked by Sidelnikov and Shestakov \cite{SS92}
(although an earlier article from Roth and Seroussi \cite{RS85}
already explained how to recover the structure of a GRS code).  To
overcome the attack of Sidelnikov and Shestakov while trying to keep
the benefits of GRS codes, several proposals appeared in the
literature. Berger and Loidreau proposed to replace the GRS code by a
subcode of low codimension \cite{BL05}; Wieschebrink \cite{W06}
included some random columns in a generator matrix of a GRS code; his
approach was enhanced in NIST submission RLCE \cite{W16,W17}, where
random columns are included and then ``mixed'' with the original
columns using specific linear transformations; finally Baldi, Bianchi,
Chiaraluce, Rosenthal and Schipani (BBCRS) \cite{BBCRS14} proposed to
mask the structure of a GRS code by right multiplying it by a
``partially weight preserving'' matrix. All these proposals have been
partially or fully broken using attacks derived from a square code
distinguisher. The first contribution, due to Wieschebrink \cite{W10}
broke the Berger-Loidreau proposal. Later on, Wieschebrink's and BBCRS
schemes were attacked in \cite{CGGOT14,COTG15} and RLCE in
\cite{CLT19}.

\medskip

In summary, forty years of research on the use of algebraic codes for
public key encryption boil down to the following observations.
\begin{enumerate}[(1)]
\item On one hand, the raw use of GRS codes as well as most of the
  variants using these codes lead to insecure schemes.
\item On the other hand, Goppa codes or more generally alternant codes
  remain robust decades after they were initially proposed by
  McEliece.
\end{enumerate}

Alternant codes are nothing but subfield subcodes of GRS codes. Hence,
considering the spectrum with (full) GRS codes on one end and their
subfield subcodes (\ie alternant codes) on the other, the intermediary
case is that of {\em subspace subcodes} of Reed--Solomon codes. This
notion of {\em subspace subcodes} of Reed-Solomon (SSRS) was
originally introduced without any cryptographic motivation by Solomon,
McEliece and Hattori \cite{S93, MS94, H95, HMS98}.  An SSRS code is a
subset of a parent Reed-Solomon code over $\F_{q^m}$ consisting of the
codewords whose components all lie in a fixed $\lambda$-dimensional
$\Fq$-vector subspace of $\F_{q^m}$, for some $\lambda \leq m$.  These
codes are no longer linear over $\F_{q^m}$ but only over $\F_q$. The
SSRS construction provides long codes with good parameters over
alphabets of moderate size, in the spirit of alternant codes
\cite[Chapter~12]{MS86}. Therefore these codes are interesting from an
information-theoretic point of view.

\medskip

For public key cryptography, some recent works exploring different
approaches appeared in the recent years.  The first use of the notion
of subspace subcodes for cryptography comes from Gabidulin and
Loidreau who propose to use subspace subcodes of Gabidulin codes for a
rank-metric based cryptosystem in \cite{GL05,GL08}.  The first work
discussing the use of SSRS codes for public key encryption is due to
Berger, Gueye and Klamti \cite{BGK19}. Shortly after, Berger, Gueye,
Klamti and Ruatta \cite{BGKR19} proposed a cryptosystem based on
quasi--cyclic subcodes of SSRS codes. In another line
of work, Khathuria, Rosenthal and Weger \cite{KRW21} proposed a
McEliece--like encryption scheme using expanded subspace subcodes of
GRS codes instead of Goppa codes. Throughout the document, we will
refer to this scheme as the {\em XGRS} scheme (where the \textit{X}
stands for \textit{expanded}).

The use of subspace subcodes of Reed--Solomon codes is of particular
interest in code based cryptography since it includes McEliece's
original proposal based on Goppa codes on the one hand and encryption
based on generalised Reed--Solomon codes on the other hand as the two
extremities of a same spectrum. Indeed, starting from Reed--Solomon
codes over $\Fqm$ and considering subspace subcodes over subspaces of
$\Fqm$ of dimension $1 \leq \lambda \leq m$, the case $\lambda = m$ is
corresponds to GRS codes, while the case $\lambda = 1$ corresponds to
alternant codes (which include Goppa codes). The notion of subspace
subcodes allows a modulation of the parameter $\lambda$. Consequently,
they are of particular interest for two reasons.

\begin{enumerate}[(1)]
\item Subspace subcodes may provide interesting codes for encryption
  with $\lambda > 1$, providing shorter keys than the original
  McEliece scheme.
\item Their security analysis encompasses that of Goppa and alternant
  codes and may help to better understand the security of McEliece
  encryption scheme. We emphasise that such a security analysis is of
  crucial interest since {\em Classic McEliece} lies among the very
  few candidates selected by the NIST for the last round of the
  post-quantum standardisation process.
\end{enumerate}

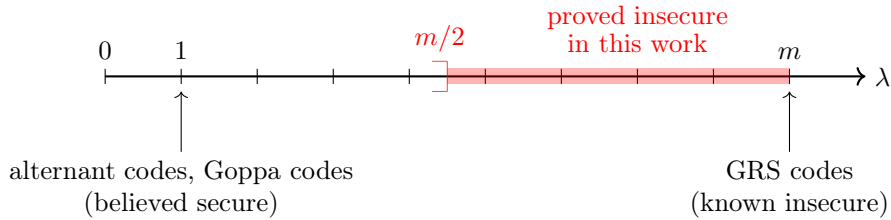
\begin{figure}[h!]
\centering
\begin{tikzpicture}[scale=1]
  \def \eps {0.1}

  \draw [black, thick=1pt, ->] (0,0) -- (10,0);
  \foreach \x in {0,...,9}
  { \draw (\x,-\eps) -- (\x,\eps); }
  \node [right] at (10,0) {$\lambda$};
  \node [above] at (0,\eps) {$0$};
  \node [above] at (1,\eps) {$1$};
  \node [above] at (9,\eps) {$m$};
  \draw [red,opacity=0.7] (4.5,-0.2) -- (4.5,0.2);
  \draw [red,opacity=0.7] (4.3,-0.2) -- (4.5,-0.2);
  \draw [red,opacity=0.7] (4.3,0.2) -- (4.5,0.2);
  \fill[red,opacity=0.3] (4.5,-0.1) rectangle (9,0.1);
  \node [above,red] at (4.4,0.2) {$m/2$};
  \node [above,red] at (7,0.5) {proved insecure};
  \node [above,red] at (7,0.2) {in this work};
  \draw [->] (1,-1) -- (1,-0.2);
  \draw [->] (9,-1) -- (9,-0.2);
  \node [below] at (1,-1) {alternant codes, Goppa codes};
  \node [below] at (1,-1.4) {(believed secure)};
  \node [below] at (9,-1) {GRS codes};
  \node [below] at (9,-1.4) {(known insecure)};
  \end{tikzpicture}
\caption{The case $\lambda=m$ (\ie using the whole GRS code as secret
  key) is known to be insecure. On the other hand, the case
  $\lambda=1$ corresponding to alternant codes (among which Goppa
  codes) is a well-studied hard problem. Our attack covers all
  cases where $m/2 < \lambda \leq m$.}
\label{fig:fleche}
\end{figure}
\subsection*{Our contribution}
In the present article, we first introduce a general public key
cryptosystem relying on subspace subcodes of Reed--Solomon codes, which
we refer to as the \emph{SSRS cryptosystem}. We prove that the XGRS
cryptosystem of \cite{KRW21} is in fact a sub-instance of the SSRS
scheme. Then, we analyse the security of the SSRS cryptosystem, using
alternatively a high-level approach (considering abstract subspaces of
$\Fqm$) or a more constructive one (focusing on explicit descriptions
of such codes as $\Fq$-linear codes using the \emph{expansion
  operator}). We present a distinguisher on the SSRS scheme, by
introducing a new and original notion, which we called \emph{twisted
  product of codes}.  Our distinguisher succeeds for any subpace
subcode of a GRS code over $\Fqm$ when the subspaces have dimension
$\lambda > m/2$ (see Figure~\ref{fig:fleche}).  Using this
distinguisher we derive a polynomial time attack on SSRS when
$\lambda > m/2$, which in particular breaks some parameters of the
XGRS scheme proposed by its authors, namely the case of subspace
subcodes on $2$--dimensional subspaces of GRS codes over $\F_{q^3}$.

\subsection*{Related work}
This work should be related to \cite[\S~VI.B]{BGK19} where it is shown
that an encryption scheme based on an expanded GRS code (the full
code, not a subspace subcode) is not secure.  Next, in the same
reference \cite[\S~VI.C]{BGK19}, the case of subspace subcodes is
discussed and the proposed attack involves a brute force search on the
all the bases used for the expansion of each position. This attack has
an exponential complexity, while ours runs in polynomial time.

On the other hand, the ``$\lambda > m/2$ condition'' for the public
code to be distinguished from a random one should be compared to the
results of \cite{COT14a, COT17}, where some classical Goppa codes are
attacked using, among others, the square code operation. The
considered classical Goppa codes correspond to the parameters $m = 2$
and $\lambda = 1$ and hence lie at the very limit of the
distinguisher.  These Goppa codes got however broken due to peculiar
features that give them a larger dimension compared to generic
alternant codes of the same parameters (see \cite{SKHN76,COT14}).

\subsection*{Outline of the article}
We start in Section~\ref{sec:nota} by fixing the
notation and bringing well--known notions of algebraic coding theory
that are necessary to follow the article. The notion of subspace
subcodes is recalled in Section~\ref{sec:SScodes} as well as some of
their known properties. We also discuss the way to practically
represent subspace subcodes, and introduce the SSRS cryptosystem,
which is an instantiation of the McEliece scheme using
subspace subcodes of Reed--Solomon codes. Some operators play an
important role in this article, Section~\ref{sec:tools} is devoted to
them and the way they interact with each other. In
Section~\ref{sec:SSRS_cryptosystem}, we present the XGRS cryptosystem
from \cite{KRW21} and prove that it is a sub-instance of the SSRS
cryptosystem. The notion of \emph{twisted square code} and the
corresponding distinguisher are introduced in
Section~\ref{sec:distinguisher}. Finally, Section~\ref{sec:attack} is
devoted to the presentation of the attack on SSRS scheme.

\medskip

\noindent {\bf Note.} A long version of the present article including
further detailed proofs and remarks is available online. See \cite{CL20}.

\subsection*{Acknowledgements} The authors express their deep
gratitude to the anonymous referees whose relevant comments and
suggestions permitted a significant improvement of the present
article. This research was conducted while Matthieu Lequesne was
employed by Sorbonne Universit\'e and Inria. Matthieu Lequesne is now
employed at the CWI in Amsterdam (Netherlands) and is funded by the
ERC-ADG-ALGSTRONGCRYPTO project (no. 740972).


\section{Notation and prerequisites}\label{sec:nota}
In this section, we fix the notation and recall some usual tools of
code-based cryptography that are used in the definition of subspace
subcodes and in the construction the XGRS cryptosystem.

\subsection{Notation}
In this article, $q$ denotes a power of a prime and $m$ a positive
integer. The vector space of polynomials of degree less than $k$ over
a field $\F$ is denoted by $\F[X]_{<k}$.  The space of matrices with
entries in a field $\F$ with $m$ rows and $n$ columns is denoted by
$\F^{m \times n}$.  Given an $\F$--vector space $\mathcal{V}$ and
vectors $\vv_0, \dots, \vv_{s-1} \in \mathcal{V}$, the subspace
spanned over $\F$ by the $\vv_i$'s is denoted by
\if\longversion1
\[
\Span{\vv_0, \dots, \vv_{s-1}}{\F} \eqdef \left\{ \sum_{i=0}^{s-1}
  \lambda_i \vv_i \,\bigg|\, \lambda_i \in \F \right\}.
\]
\else
$\Span{\vv_0, \dots, \vv_{s-1}}{\F}$.
\fi
Moreover, given a matrix $\Gm \in \F^{k \times n}$, we denote by
$\Span{\Gm}{\F}$ the space spanned by the {\bf rows} of $\Gm$, that is
to say the code with generator matrix $\Gm$, is denoted by
\if\longversion1
\[
  \Span{\Gm}{\F} \eqdef \Span{\vv ~|~ \vv \ {\rm row\ of\ }\Gm}{\F}
  = \{\mv \cdot \Gm ~|~ \mv \in \F^k\}.
\]
\else
$\Span{\Gm}{\F}$.
\fi
The {\em support} of a vector is the set of its indices with non-zero
entries. An $[n,k]$ code over $\F$ is a linear code over $\F$ of
length $n$ and dimension $k$.

Given two integers $a, b$ with $a < b$, we denote by $\IInt{a}{b}$ the
interval of integers $\{a, a+1, \dots, b\}$.
The cardinality of a finite set $U$ is
denoted by $|U|$.  Given a probabilistic event $A$, its probability is
denoted by $\Prob A$\if\longversion1 and the mean of a random variable $X$ is denoted by $\Esp X$.
\else.
\fi

\if\longversion1
\medskip

\noindent {\bf Convention.} In this article, any word or finite
sequence of length $\ell$ is indexed from $0$ to $\ell - 1$. In
particular, codewords of length $n$ are indexed as follows: $(x_0,
\dots, x_{n-1})$.
\fi

\subsection{Reed--Solomon codes}

\begin{definition}[Generalised Reed--Solomon codes]\label{def:GRS} Let
$\xv \in \F_q^n$ be a vector whose entries are pairwise distinct and
$\yv \in \F_q^n$ be a vector whose entries are all nonzero. The {\em
generalised Reed--Solomon (GRS) code with support $\xv$ and multiplier
$\yv$ of dimension $k$} is defined as
\[ \GRS{k}{\xv}{\yv} \eqdef \left\{(y_0 f(x_0), \ldots, y_{n-1} f(x_{n-1}))
~|~ f \in \F_q[x]_{<k}\right\}.
\]
If $\yv = (1,\dots, 1)$ then the code is said to be a {\em Reed--Solomon}
code and denoted as $\RS{k}{\xv}$.
\end{definition}

\subsection{Component wise product of codes}
\label{subsec:starprod}
The component-wise product of two vectors $\av$ and $\bv$ in $\F^n$ is
denoted by
\[ \av \,\star\, \bv \eqdef (a_0b_0, \ldots, a_{n-1} b_{n-1}).
\] This definition extends to the product of codes, where the {\em
  component-wise product} or {\em $\star$--product} of two
$\mathbb{K}$-linear codes $\code A$ and $\code B \subseteq \F^n$
spanned over a field $\mathbb{K} \subseteq \F$ is defined as
\[ \code A \star_{\mathbb{K}} \code B \eqdef \Span{ \av
\,\star\, \bv ~|~ \av \in \code A, \ \bv \in \code B }{\mathbb{K}}.
\] When $\code A = \code B$, we denote by
$\sq{\code A}_{\mathbb{K}} \eqdef \code A \star_{\mathbb{K}} \code A$
the {\em square} code of $\code A$ spanned over $\mathbb{K}$.

\begin{remark}
  The field $\mathbb{K}$ in the above notation is almost always
  equal to $\F$ the base field on which the codes are defined. However, it may
  sometimes be a subfield. For the sake of clarity, we make the value
  of $\mathbb{K}$ explicit only in the ambiguous cases.
\end{remark}

\subsection{Punctured and shortened codes}
\label{subsec:shortening}

The notions of \textit{puncturing} and \textit{shortening} are
classical ways to build new codes from existing ones. These
constructions will be useful for the attack. We recall here their
definition. For a codeword $\cv \in \F_q^n$, we denote by
$(c_0,\ldots,c_{n-1})$ its entries.

\begin{definition}[Puncturing, shortening] Let $\CC \subseteq \F_q^n$
  and $\Lcal \subseteq \IInt{0}{n-1}$.  The {\em puncturing} and the
  {\em shortening of $\CC$ at $\Lcal$} are defined as
  \[ \pu{\CC}{\Lcal} \eqdef \{(c_i)_{i \in \IInt{0}{n-1} \setminus
      \Lcal}\ ~|~ \cv \in \CC\} \quad {\rm and} \quad
    \sh{\CC}{\Lcal} \eqdef \pu{\{\cv \in \CC\ ~|~ \forall i
      \in \Lcal,\ c_i = 0\}}{\Lcal}.
  \]
\end{definition}

Similarly, given a matrix $\Mm$ with $n$ columns, one defines
$\pu{\Mm}{\Lcal}$ as the matrix whose columns with index in $\Lcal$
are removed, so that puncturing a generator matrix of a code yields a
generator matrix of the punctured code.

\if\longversion1
Shortening a code is equivalent to puncturing the dual code, as
explained by the following proposition.
\fi

\begin{proposition}[{\cite[Theorem 1.5.7]{HP03}}]\label{prop:dual_pu_sh}
  Let $\CC$ be a linear code
over $\F_q^{n}$ and $\Lcal \subseteq \IInt{0}{n-1}$. Then,
\[ \sh{\Dual(\CC)}{\Lcal} = \Dual(\pu{\CC}{\Lcal}) \quad \text{ and } \quad
\Dual(\sh{\CC}{\Lcal}) = \pu{\Dual(\CC)}{\Lcal},
  \] where $\Dual(\AC)$ denotes the dual of the code $\AC$.
\end{proposition}

\begin{remark} The notation $\Dual(\CC)$ to denote the dual code of
  $\CC$ is rather unusual. This code is commonly denoted
  $\CC^{\perp}$.
\end{remark}

\subsection{Bases and trace map}
In this paper, we consider codes over a finite field $\Fqm$ and
codes over its subfield $\Fq$.
A useful and natural tool is the trace
\if\longversion1
map.
\begin{definition}[Trace map] Let $q$ be a prime power and $m$ an integer. The
trace map is defined as
\[ \Tr: \qquad \map{\Fqm}{\Fq}{x}{
\sum_{i=0}^{m-1}x^{q^i}.}
\]
\end{definition}

\begin{definition}[{\cite[Definition 2.30]{LN97}}]\label{def:dual_basis} Let $\B = (b_0,
  \ldots, b_{m-1})$ be an $\Fq$-basis of $\Fqm$. There exists a unique
  basis $\Bdual = (b^{\ast}_0, \ldots, b^{\ast}_{m-1})$, such that :
\[ \forall 0 \leq i, j \leq m-1, \qquad \Tr(b_ib^{\ast}_j) = \left\{
\begin{array}{c l} 1 & \text{if } i=j,\\ 0 & \text{otherwise}.
\end{array} \right.
\] This basis will be referred to as the \emph{dual} basis of $\B$ and
denoted $\Bdual$.
\end{definition}

Given an $\Fq$--basis $\B = (b_0, \ldots, b_{m-1})$ of
$\Fqm$ and $x$ an element of $\Fqm$. Then
\else
map denoted $\Tr : \Fqm \rightarrow \Fq$.
Recall that, from {\cite[Definition 2.30]{LN97}},
to any $\Fq$--basis $\B = (b_0, \ldots, b_{m-1})$
corresponds a unique {\em dual basis} $\Bdual = (b^{\ast}_0, \ldots,
b^{\ast}_{m-1})$ such that
for any $x \in \Fqm$,
\fi
 the expression of $x$
as an $\Fq$--linear combination of the elements of $\B$ writes as
\if\longversion1
\begin{equation*}
  x = \Tr(b^*_0 x) b_0 + \cdots + \Tr(b_{n-1}^* x)b_{n-1},
\end{equation*}
where $\Bdual = (b^{\ast}_0, \ldots,
b^{\ast}_{m-1})$ denotes the dual basis of $\B$.
\else
\begin{equation*}
  x = \Tr(b^*_0 x) b_0 + \cdots + \Tr(b_{n-1}^* x)b_{n-1}.
\end{equation*}
\fi


\section{Subspace Subcodes}\label{sec:SScodes}

\subsection{Definition and first properties}
\begin{definition}[\cite{HMS98}] Given a linear code $\CC$ defined
over a field $\Fqm$, and a $\lambda$-dimensional subspace $\SS$ of
$\Fqm$ ($0 \leq \lambda \leq m$), the subspace subcode $\CCSS$ is
defined to be the set of codewords of $\CC$ whose components all lie
in $\SS$.
  \[ \CCSS \eqdef \{\cv \in \CC ~|~ \forall i \in \IInt{0}{n-1},\ c_i
\in \SS\} \subseteq \Fqm^n.
  \]
\end{definition}

It is important to note that the code $\CCSS$ is an $\Fq$--linear
subspace of $\Fqm^n$ which is generally neither $\Fqm$--linear nor
linear over some intermediary extension.  Since each entry of a
codeword can be represented as $\lambda$ elements of $\Fq$, the code
could be converted into a code over the alphabet $\Fq^{\lambda}$. Such
a code would form an additive subgroup over ${(\Fq^\lambda)}^n$,
therefore this construction is called a \textit{subgroup subcode} by
Jensen in \cite{J95}.  In a context of message transmission, this
natural way to represent such a subspace subcode is detailed further
in \S~\ref{subsec:expansion}.

We can generalise this definition with different subspaces for each
entry.

\begin{definition}
\label{def:subspace-generalised} Given a linear code $\CC$ of length
$n$ over a field $\Fqm$, and the $\lambda$-dimensional subspaces
$(\SS_0, \ldots, \SS_{n-1})$ of $\Fqm$ ($0 \leq \lambda \leq m$), the
subspace subcode $\CCSSi$ is defined to be the set of codewords of
$\CC$ such that the $i$-th components lies in $\SS_i$.

  \[ \CCSSi \eqdef \{\cv \in \CC ~|~ \forall i \in \IInt{0}{n-1},\ c_i
\in \SS_i\}.
  \]
\end{definition}

\if\longversion1
\begin{remark} When $\SS_0 = \cdots = \SS_{n-1} = \Fq$, then we find
the usual definition of {\em subfield subcode}.
\end{remark}

\begin{remark}
  It is possible to give a more general definition where the $\SS_i$'s
  do not have the same dimension $\lambda$. However, such a broader
  definition would be useless in the present article.
\end{remark}
\fi

\begin{proposition}
  \label{prop:SS_parameters} Let $\CC$ be a linear code of length $n$
and dimension $k$ over $\Fqm$ and $\SS \subseteq \Fqm$ be a subspace
of dimension $\lambda \leq m$. Then
\begin{equation}\label{eq:dim_SScode}
  \dim_{\Fq} \CCSS \geq km - n(m-\lambda).
\end{equation}
\end{proposition}

\begin{proof} See for instance \cite[Theorem 2 (1)]{J95}.
\end{proof}

\if\longversion1
\begin{example}
  The following example comes from \cite{HMS98}.
  
Consider $\CC$ the Reed--Solomon code over $\F_{2^4}$ of length 15 and
dimension 9. This code has minimum distance 7. Any element of
$\F_{2^4}$ can be decomposed over the $\F_2$-basis $(1, \alpha,
\alpha^2, \alpha^3)$, where $\alpha$ is a root of the irreducible
polynomial $X^4+X+1$. Let $\SS$ be the subspace spanned by $(1,
\alpha, \alpha^2)$. The code $\CCSS$ is the subset of codewords of
$\CC$ that have no component in $\alpha^3$. Hence, if one uses this
code for communication, there is no need to send the $\alpha^3$
component, since it is always zero.

So this subspace subcode can be seen as an $\F_2$-linear code of
length 15 over the set of binary 3-tuples. But the code is not a
linear code over $\F_{2^3}$. The minimum distance of $\CCSS$ is at
least 7, because it cannot be less than the minimum distance of the
parent code. The number of codewords in $\CCSS$ is $2^{22}$.
As a comparison, one other way to create a code of length 15 over
binary 3-tuples is by shortening the generalised BCH code $[63,52,7]$
over $\F_{2^3}$. This gives a $[15,4,\geq 7]$ code over $\F_{2^3}$
which has $2^{12}$ codewords. 
\end{example}
\fi

Similarly to the case of subfield subcodes, inequality
(\ref{eq:dim_SScode}) is typically an equality as explained in the
following statement that we prove because of a lack of references.

\begin{proposition}\label{prop:typical_dim_SS}
  Let $\RC$ be a uniformly random code among the codes of length $n$
  and dimension $k$ over $\Fqm$. Let $\SS_0, \dots, \SS_{n-1}$ be
  $\Fq$--subspaces of $\Fqm$ of dimension $\lambda$.  Suppose that
  $km > n(m-\lambda)$.  Then, for any integer $\ell$, we have
  \[
    \Prob{\dim_{\Fq}\SScode{\RC}{(\SS_0, \dots, \SS_{n-1})} \geq km -
      n(m-\lambda) + \ell} \leq q^{-\ell} \left(
      \frac{1}{1-q^{-mn}} + \frac{1}{q^{km - n(m - \lambda)}}
    \right).
  \]
  In particular, for fixed values of $q, m$ and $\lambda$, this
  probability is in $O(q^{-\ell})$ when $n \to \infty$.
\end{proposition}

\if\longversion1
\begin{proof}
  Let $\Grand$ be a uniformly random
  variable among the full rank matrices in $\Fqm^{k \times n}$ and
  \[
    \RC \eqdef \{\mv \Grand ~|~ \mv \in \Fqm^k\}.
  \]
  The code $\RC$ is uniformly random among the set of $[n,k]$ codes
  over $\Fqm$ (\cite[Lemma~3.12]{C20}).
  Let $\Phi$ be the $\Fq$--linear canonical projection
  \[
    \Phi : \Fqm^n \longrightarrow \prod_{i=0}^{n-1}\Fqm/\SS_i.
  \]
  Then, $\SScode{\RC}{(\SS_0, \dots, \SS_{n-1})}$ is the kernel of
  the restriction of $\Phi$ to $\RC$ and hence,
  \begin{align}
    \nonumber\Esp{|\SScode{\RC}{(\SS_0, \dots, \SS_{n-1})}|} &=
    \Esp{\sum_{\mv \in \Fqm^k} \mathds{{1}}_{\Phi(\mv \Grand) = 0}}\\
    \nonumber
    & = \sum_{\mv \in \Fqm^k} \Prob{\Phi(\mv \Grand) = 0} 
    \\
    \label{eq:bound_esp}
    & = 1 + \sum_{\mv \in \Fqm^k\setminus{\{0\}}} \Prob{\Phi(\mv \Grand) = 0}.
  \end{align}
  Since $\Grand$ is uniformly random among the full--rank matrices,
  then for any $\mv \in \Fqm^k \setminus \{0\}$, the vector
  $\mv \Grand$ is uniformly random in $\Fqm\setminus \{0\}$
  (\cite[Lemma~3.13]{C20}) and hence
  \begin{align*}
    \forall \mv \in \Fqm^k \setminus \{0\}, \quad \Prob{\Phi (\mv
    \Grand = 0)}& = \frac{\left|\ker \Phi \setminus \{0\}
                  \right|}{\left|\Fqm^n \setminus \{0\}\right|}\\
                &= \frac{\left|\prod_i\SS_i\right| - 1}{q^{mn}-1}\\
                &= \frac{q^{\lambda n} - 1}{q^{mn}-1} \leq q^{-n(m-\lambda)} \cdot \frac{1}{1-q^{-mn}}\cdot
  \end{align*}
  Thus, applied to \eqref{eq:bound_esp},
  \begin{align*}
    \Esp{|\SScode{\RC}{(\SS_0, \dots, \SS_{n-1})}|} & \leq 1 +
     |\Fqm^k \setminus \{0\}| \cdot q^{-n(m-\lambda)} \cdot \frac{1}{1-q^{-mn}} \\
     &\leq  1 + q^{km - n(m-\lambda)} \cdot \frac{1}{1-q^{-mn}}\cdot
  \end{align*}
  Finally, using Markov inequality, we get
  \begin{align*}
    \Prob{\dim_{\Fq}(\SScode{\RC}{(\SS_0, \dots, \SS_{n-1})}) \geq km - n(m-\lambda) + \ell} 
    &= \Prob{|\SScode{\RC}{(\SS_0, \dots, \SS_{n-1})}| \geq
      q^{km-n(m-\lambda) + \ell}}\\
    &\leq \frac{\Esp{|\SScode{\RC}{(\SS_0, \dots, \SS_{n-1})}|}}{q^{km-n(m-\lambda) + \ell}}\\
    & \leq q^{-\ell} \left( \frac{1}{1-q^{-mn}} +
      \frac{1}{q^{km - n(m - \lambda)}} \right).
  \end{align*}
\end{proof}
\else
\begin{proof}
  See \cite{CL20}.
\end{proof}
\fi

\subsection{How to represent subspace
  subcodes?}\label{subsec:expansion}

For a practical implementation, subspace subcodes may be represented
as codes over the subfield $\Fq$ with a higher length. For this sake
we introduce the {\em expansion operator} and give some of its
properties.

\subsubsection{The Expansion Operator}

\begin{definition}[Expansion of a vector or a code] For an $\Fq$--basis $\B$
  of $\Fqm$, let $\ExpVect_{\B}$ denote the expansion of a vector over
  the basis $\B$ defined by
\[ \ExpVect_{\B} : \map { \Fqm^\ell } { \Fq^{m\ell} } { (x_0, \ldots, x_{\ell-1}) } {
(x_{0,0}, \dots, x_{0, m-1}, \dots, x_{\ell-1,0}, \dots, x_{\ell - 1, m-1}),}
\]
where for any $i \in \IInt{0}{\ell-1}$,
$x_i = x_{i,0}b_0 + \cdots + x_{i,m-1}b_{m-1}$.
Given a linear code $\CC$ of
  length $n$ over $\Fqm$, denote
  $\ExpCode_{\B}(\CC)$ the linear code over $\Fq$ defined by
  \[ \ExpCode_\B(\CC) \eqdef \{ \ExpVect_{\B}(\cv) \,|\, \cv \in
  \CC\}.\]
We will
apply this operator to vectors or codes of different lengths $\ell$.
\end{definition}

Note that if $\Bdualall$
denotes the dual basis of $\B$ then
\[ 
  \ExpVect_{\B}(\xv) = (\Tr(b^{\ast}_0x_0), \ldots,
  \Tr(b^{\ast}_{m-1}x_0), \,\ldots\,, \Tr(b^{\ast}_0x_{\ell-1}),
  \ldots, \Tr(b^{\ast}_{m-1}x_{\ell - 1})).
\]

We also wish to define the expansion of a matrix so that the expansion
of a generator matrix of a code $\CC \subseteq \Fqm^n$ over a basis
$\B$ is a generator matrix of $\ExpCode_{\B}(\CC)$.  It turns out that
expanding the rows of a generator matrix is not sufficient. Indeed,
for a given codeword $\cv \in \CC$ and any $\alpha \in \Fqm$, the
vector $\ExpVect_\B(\alpha \cv)$ is in $\ExpCode_\B (\CC)$.  This
leads to the following definition.

\begin{definition}[Expansion of a matrix] Given $\Ball$ an $\Fq$-basis
  of $\Fqm$. Let $\ExpMat_{\B}$ denote the following operation.
  \[
    \ExpMat_\B :
    \map { \Fqm^{k \times n} } { \Fq^{mk \times mn} } {
      \left(

        \begin{array}{ccc} & \mv_{0} & \\ && \\
                           & \vdots & \\
          && \\
& \mv_{k-1} & \\
\end{array} \right)
} {
\left(\begin{array}{ccccc}
  & \ExpVect_{\B}( & b_0\mv_{0} & ) & \\
  && \vdots && \\
  & \ExpVect_{\B} (& b_{m-1}\mv_{0} &) & \\
  && \vdots && \\
  & \ExpVect_{\B} (& b_0\mv_{k-1} & ) & \\
  && \vdots && \\
  & \ExpVect_{\B} ( & b_{m-1}\mv_{k-1} & ) & 
\end{array}\right).
}
\]
\end{definition}

\begin{remark}
  Caution, applying $\ExpMat_\B$ to an $1 \times n$ matrix returns an
  $m \times mn$ matrix. It is {\bf not} equivalent to applying
  $\ExpVect_\B$ to the vector corresponding to this row.
\end{remark}

\if\longversion1
\begin{remark}
  $ \ExpMat_{\Bdual}(\Mm) = \transposeb{\ExpMat_{\B}(\transpose{\Mm})}.$
\end{remark}
\fi

\begin{proposition}[{\cite[Proposition 1]{KRW21}}]
  \label{prop:lift} Let $\CC$ be a linear code of dimension $k$ and
length $n$ over $\Fqm$. Let $\Gm$ denote a generator matrix of $\CC$
and $\Hm$ denote a parity-check matrix of $\CC$. Then, for any fixed
$\Fq$-basis $\B$ of $\Fqm$, the following hold.
\begin{enumerate}[(i)]
\item\label{item:1st_statement} For all $\xv \in \Fqm^k$, we have
  $\ExpVect_\B(\xv\,\cdot\,\Gm) = \ExpVect_\B(\xv) \,\cdot\,
  \ExpMat_\B(\Gm)$.
\item\label{item:2nd_statement} For all $\yv \in \Fqm^n$, we have
  $\transpose{\ExpVect_\B(\transpose{(\Hm \,\cdot\,
      \transpose{\yv})})} = \ExpMat_{\Bdual}(\Hm) \,\cdot\,
  \transpose{\ExpVect_\B(\yv)}$.
\end{enumerate}
\end{proposition}

\begin{cor}\label{cor:HmGmExp}
  Let $\Gm$ and $\Hm$ be a generator and a parity-check matrix of
  $\CC$ respectively. Let $\B$ denote an $\Fq$-basis of $\Fqm$.  Then
  $\ExpMat_\B(\Gm)$ and $\ExpMat_{\Bdual}(\Hm)$ are respectively a
  generator matrix and a parity-check matrix of $\ExpCode_{\B}(\CC)$.
\end{cor}

\begin{definition}[Block]\label{def:block}
  Given a vector $\vv \in \Fqm^n$, an $\Fq$--basis $\B$ of $\Fqm$ and
  a non negative integer $i < n$, the $i$--th {\em block} of the
  expanded vector $\ExpVect_{\B}(v) \in \Fq^{mn}$ is the length $m$
  vector composed by the entries of index $mi, mi+1, \dots, mi+m-1$ of
  $\ExpVect_{\B}(v)$. It corresponds to the decomposition over $\B$ of
  the $i$-th entry of $\vv$. We extend this definition to matrices,
  where the $i$-th block of an expanded matrix means the $mk \times m$
  matrix whose rows correspond to the $i$-th block of each row of the
  expanded matrix. In particular, the expansion in a basis $\B$ of
  some $\xv \in \Fqm^n$ is the concatenation of $n$ blocks of length
  $m$.
\end{definition}

\subsubsection{Expansion over various bases}

We have seen in Definition~\ref{def:subspace-generalised} that we
could define a subspace subcode with different subspaces for each
entry. Similarly, we can define an expansion with regard to a different
basis for each entry.

\begin{definition} Given $\ell$ bases $(\B_0,\ldots, \B_{\ell-1})$ of
  $\Fqm$, we define $\ExpVect_{{(\B_i)}_i}(\xv)$ as the expansion of
  $\xv \in \Fqm^\ell$, whose $i$\textsuperscript{th} block is the
  expansion of $x_i$ over the basis $\B_i$.  For a linear code
  $\CC \subseteq \Fqm^n$, we
  define $\ExpCode_{{(\B_i)}_i}(\CC)$ in a similar fashion.
Finally, given a matrix $\Mm \in \Fqm^{k\times n}$ and an additional $\Fq$--basis
$\bar{\B} = (\bar{b}_0, \dots , \bar{b}_{m-1})$ of $\Fqm$, we define
\[
  \ExpMat_{{(\B_i)}_i}^{\bar{\B}}(\Mm) \eqdef
\left(\begin{array}{ccccc}
  & \ExpVect_{{(\B_i)}_i}( & \bar{b}_0\mv_{0} & ) & \\
  && \vdots && \\
  & \ExpVect_{{(\B_i)}_i} (& \bar{b}_{m-1}\mv_{0} &) & \\
  && \vdots && \\
  & \ExpVect_{{(\B_i)}_i} (& \bar{b}_0\mv_{k-1} & ) & \\
  && \vdots && \\
  & \ExpVect_{{(\B_i)}_i} ( & \bar{b}_{m-1}\mv_{k-1} & ) & 
\end{array}\right).  
\]
\end{definition}

\if\longversion1
Using dual bases we get an explicit description:
\[\ExpVect_{(\B_i)_i} (\xv) = 
(\Tr(b^{\ast}_{0,0}x_0), \ldots, \Tr(b^{\ast}_{0,m-1}x_0), \,\ldots\,,
\Tr(b^{\ast}_{\ell-1,0}x_{\ell-1}) \ldots,
\Tr(b^{\ast}_{\ell-1,m-1}x_{\ell - 1})),
\] where $\B_i = (b_{i,0}, \ldots, b_{i,m-1})$.
\fi

The properties of Proposition~\ref{prop:lift}
still hold for various bases.

\begin{remark} Note that contrary to the expansion of codes, the
  expansion of a matrix depends on the choice of a basis $\bar{\B}$ for
  the vertical expansion. When considering the code spanned by an
  expansion matrix, different choices of $\bar{\B}$ yield the same
  code, so we will omit the vertical expansion basis in the expansion
  matrix operator.
\end{remark}

\subsubsection{Squeezing: the inverse of expansion}
\label{sub:squeeze}

We can define the ``inverse'' of the expansion operator.

\begin{definition}[Squeezing] Let $\Ball$ be a basis of
$\Fqm$. Let $\xv =
(x_{0,0},\ldots,x_{0,m-1},\ldots,x_{n-1,0},\ldots,
x_{n-1,m-1}) \in \Fq^{mn}$.  We define the \emph{squeezed vector}
of $\xv$ with respect to the basis $\B$ as
  \[\SqueezeVect_\B(\xv) \eqdef \bigg( \sum_{j=0}^{m-1} x_{0,j}
b_j, \ldots, \sum_{j=0}^{m-1} x_{n-1,j}b_j \bigg) \in \Fqm^n.\]
  Let $\CC$ be an $[m \times n, k]$--code over $\Fq$.  We define the
\emph{squeezed code} of $\CC$ with respect to $\B$ as
  \[\SqueezeCode_\B(\CC) \eqdef \left\{ \SqueezeVect_\B(\cv) \,|\, \cv \in \CC
    \right\}.\] Finally, given a matrix $\Mm \in \Fq^{mk \times mn}$,
  then $\SqueezeMat_\B(\Mm) \in \Fq^{mk \times n}$ denotes the matrix
  whose rows are obtained by squeezing each row of the matrix $\Mm$
  over $\B$.
\end{definition}

\begin{remark}
  Beware that if $\Mm$ is obtained by expanding a matrix of rank $r$ over the
  basis $\B$, then $\SqueezeMat_{\B}(\ExpMat_{\B}(\Mm))$ is a
  $km \times n$ matrix and hence is {\bf not} equal to $\Mm$ but
  it generates the same code.
\end{remark}

\begin{proposition} Let $\CC$ be an $[n,k]$ code over $\Fqm$ and $\B$ be
a basis of $\Fqm$. Then,
  \[ \SqueezeCode_\B(\ExpCode_\B(\CC)) = \CC. \]
\end{proposition}

Similarly to the expansion operators, we can define the squeezing
operators with a different basis for each block.

\subsubsection{Representation of subspace subcodes}
\label{ss:representation_SSCodes}

Let $\CC$ be a code of length $n$ and dimension $k$ over the field
$\Fqm$ and $\SS$ denote an $\Fq$-subspace of $\Fqm$ of dimension
$\lambda \leq m$. Let
$\BS = (b_0, \ldots, b_{\lambda-1}) \in \Fqm^\lambda$ be an
$\Fq$-basis of $\SS$. Then any vector
$\cv = (c_0, \dots, c_{n-1}) \in \SS^n$, \ie whose entries are all in
$\SS$ can be expanded as
\[
  \ExpVect_{\BS}(\cv) \eqdef (c_{0,0}, \dots, c_{0,\lambda-1}, \dots, c_{n-1,0}, \dots,
  c_{n-1,\lambda-1}),
\]
where the ${c_{i,j}}'s$ are the coefficients of
the decomposition of $c_i$ in the $\BS$.
\if\longversion1
\begin{remark}\label{rem:well-def_for_SS}
  Note that the previous definition makes sense only
  for vectors in $\SS^n$.
\end{remark}
\fi
Next, the subspace subcode $\SScode{\CC}{\SS}$
can be represented as
\[
\ExpCode_{\BS}(\SScode{\CC}{\SS}) \eqdef \{\ExpVect_{\BS}(\cv) ~|~ \cv
\in \SScode{\CC}{\SS}\}.
\]
\if\longversion1
Here again, as noticed in Remark~\ref{rem:well-def_for_SS}, the notion
is well--defined only for codes with entries in $\SS$.

Similarly to
Definition~\ref{def:block}, a {\em block} refers to a set of the form
$\IInt{i\lambda}{(i+1)\lambda -1}$. That is to say, a set of
$\lambda = \dim \SS$ consecutive indexes of the expanded code,
corresponding to the decomposition of a single entry in $\SS$ in the
basis $\BS$.
\fi

\subsection{Subspace-subcodes of Reed-Solomon codes}

Expanding codes, in particular Reed--Solomon codes, over the base
field has been studied since the 1980's. For instance, in \cite{KL85,
  KL88}, Kasami and Lin investigate the weight distribution of
expanded binary Reed--Solomon codes. Sakakibara, Tokiwa and Kasahara
extend their work to $q$-ary Reed--Solomon codes \cite{STK89}.

But the idea behind subspace subcodes, which consists in keeping only
the subset of codewords that are defined over a subspace of the field,
first appears in a paper by Solomon \cite{S93c}. In a joint work with
McEliece \cite{MS94}, they define the notion of \emph{trace-shortened}
codes, which is a special case of subspace subcodes where
$\lambda = m-1$ and where the considered subspace $\SS$ is the kernel
of the trace map. These articles focus uniquely on Reed--Solomon
codes. Still, this point of view turns out to be the most general one
since a subspace subcode of a GRS code can always be regarded as a
subspace subcode of an RS code by changing the subspaces as explained
by the following statements.

\begin{proposition}
  \label{prop:scalar}
  Let $\CC \subseteq \Fqm^n$, $\SS_0, \dots, \SS_{n-1}$
  be $\Fq$--subspaces of $\Fqm$ and $\av \in {(\Fqm^\times)}^n$.
  Then,
  \[
    \SScode{(\CC \star \av)}{(\SS_0, \dots, \SS_{n-1})}
    = \SScode{\CC}{(a_0^{-1}\SS_0, \dots , a_{n-1}^{-1} \SS_{n-1})} \star \av.
  \]
  Moreover, let $\xv, \yv \in \Fqm^n$ be a support and a multiplier, then
  \[
    \SScode{ \GRS{k}{\xv}{\yv}}{(\SS_0, \dots, \SS_{n-1})} =
    \SScode{\RS{k}{\xv}}{(y_0^{-1}\SS_0, \dots, y_{n-1}^{-1}\SS_{n-1})} \star \yv.
  \]
\end{proposition}

The notion is then generalised to any kind of subspace and any
code by Jensen in \cite{J95} under the name \emph{subgroup subcodes}.
In his thesis \cite{H95} and in \cite{HMS98}, Hattori studies the
dimension of subspace subcodes of Reed--Solomon codes. Some
conjectures of Hattori are later proved by Spence in \cite{S04b}. Cui
and Pei extend the results to generalised Reed--Solomon codes in
\cite{JJ01}.  Then, Wu proposes a more constructive approach of these
codes using the equivalent of the expansion operator in \cite{W11}.
The idea of using distinct bases for each position is introduced in
\cite{vDT99}.

\subsection{An instantiation of McEliece with SSRS codes}
\label{sec:generic-scheme}

Let us first present a generic encryption scheme based on subspace
subcodes of GRS codes. This cryptosystem will be referred to as the
\emph{Subspace Subcode of Reed--Solomon} (SSRS) scheme. We will later
prove that the cryptosystem of \cite{KRW21} is a sub-instance of the
SSRS scheme.

\subsubsection{Parameters}

The cryptosystem is publicly parametrised by:
\begin{itemize}
\item $q$ a prime power;
\item $m$ an integer;
\item $\lambda$ such that $0 < \lambda < m$;
\item $n,k$ such that $0 \leq k < n \leq q^m$ and $km > (m -
  \lambda)n$.
\end{itemize}

\subsubsection{Key generation}
\label{sec:SSRS-keygen}

\begin{itemize}
\item Generate a uniformly random vector $\xv \in \Fqm^n$ with
  distinct entries.
\item Choose $n$ uniformly random $\lambda$--dimensional vector
  subspaces $\SS_0, \dots, \SS_{n-1} \subseteq \Fqm$
  with respective bases $\B_{\SS_0}, \dots, \B_{\SS_{n-1}}$.
\item Let $\Gpub \in \Fq^{(km - n(m - \lambda)) \times \lambda n}$
  denote a generator matrix of the code
  \[
    \Cpub \eqdef \ExpCode_{(\B_{\SS_0}, \dots,
      \B_{\SS_{n-1}})}\left(\RS{k}{\xv}_{|(\SS_0, \ldots,
        \SS_{n-1})}\right).
  \]
  If $\Gpub$ is not full-rank, abort and restart the process.  See
  Section~\ref{subsec:practical_construction} for the practical
  computation on $\Gpub$.
\item The public key is $\Gpub$ and the secret key is
  $(\xv, \B_{\SS_0}, \ldots, \B_{\SS_{n-1}})$.
\end{itemize}

\begin{lem}[Public Key Size] The public key is a matrix of size
  $m (n-k) \times \lambda n$ over $\Fq$. Only the non-systematic part
  is transmitted. Hence the public key size in bits is
  \[m(n-k)(\lambda n - m(n-k))\log_2(q).\]
\end{lem}

\subsubsection{Encryption}
\label{sec:SSRS-enc}

Let $\mv \in \Fq^{mk-(m-\lambda)n}$ be the plaintext. Denote
\[t \eqdef \lfloor \frac{n-k}{2} \rfloor.\]
Choose
$\ev \subseteq \Fq^{(m-\lambda) n}$ uniformly at random among vectors
of $\Fq^{(m-\lambda) n}$ with exactly $t$ non-zero blocks (see
Definition~\ref{def:block}).

\subsubsection{Decryption}
\label{sec:SSRS-dec}
From $\yv \in \Fq^{\lambda n}$, construct a vector $\yv' \in \Fq^{mn}$
by completing each block of size $\lambda$ with $m-\lambda$ entries
set to zero.

Denote
\(
  \yv'' = \SqueezeVect_{(\B_i)_i}(\yv').
\)
According to the definition of $\ev$, the vector $\yv'' \in \Fqm^n$ is
at distance $t$ of the code $\RS{k}{\xv}$. Hence, by decoding, one
computes the unique $\cv \in \RS{k}{\xv}$ at distance $\leq t$ from
$\yv''$ and the expansion of $\cv$ yields $\mv \Gpub$.


\section{Further properties of the expansion operator}
\label{sec:tools}

We now introduce some properties of the expansion operators. More
specifically, in order to analyse the XGRS cryptosystem, we study how
this operator behaves with respect to other operations (especially
those used in the key generation): puncturing/shortening, computing
the dual, changing the expansion basis. We also consider the relation
with the square product operation, as this is a natural distinguisher
for GRS-based codes.

In this section, for the sake of clarity, all properties will be
defined considering the same basis for each entry, but everything
works exactly the same way if one considered expansion with a different
basis for each entry, as different columns of $\Fqm$ (or blocks of columns
of $\Fq$ corresponding to the expansion of same column of $\Fqm$) do
not interact.

\subsection{Subspace subcodes as shortening of expanded codes}

This lemma explains how to construct the parity-check matrix of a
subspace subcode from the parity-check matrix of the parent code. This
result is important to perform computations over the subspace
subcodes.

\begin{lem}\label{lem:JLM}
  For integers $n$ and $\lambda < m$, denote $\JLM$ the
  subset of $\IInt{0}{mn-1}$ consisting of the last $m-\lambda$
  entries of each block of length $m$
  \begin{equation}\label{eq:def_JLM}
    \JLM \eqdef \left\{im+j,\, i \in \IInt{0}{n-1}, \, j
      \in \IInt{\lambda}{m-1}\right\}.
  \end{equation} 
  Let $\Ball$ be an $\Fq$-basis of $\Fqm$ such that
  $\BS = (b_0, \ldots, b_{\lambda-1})$ is a basis of $\SS$. Then,
  \[
    \ExpCode_{\BS}(\CCSS) = \sh{\ExpCode_\B(\CC)}{\JLM}.
  \]
  Equivalently, the following diagram is commutative.
  \begin{center}
  \begin{tikzcd}
  \CC
  \ar[r, "(\,\cdot\,)_{|\SS}"]
  \ar[d, "\mathbf{Exp}_{\B}"]
  & [5em]
  \CCSS
  \ar[d, "\mathbf{Exp}_{\BS}"]
  \\ [3em]
  \ExpCode_{\B}(\CC)
  \ar[r, "\sh{\,\cdot\,}{\JLM}"]
  &
  \ExpCode_{\BS}(\CCSS)
  \\
\end{tikzcd}
\end{center}

\end{lem}

\if\longversion1
Let $\Hm \in \Fqm^{k \times n}$ denote a
parity--check matrix of $\CC$.  Complete the basis
$\BS=(b_0, \ldots, b_{\lambda-1})$ with $m - \lambda$ additional
elements $(b_\lambda, \ldots, b_{m-1}) \in \Fqm^{m-\lambda}$ such that
$\B = (b_0, \ldots, b_{m-1})$ forms an $\Fq$--basis of $\Fqm$.
According to Corollary~\ref{cor:HmGmExp}, the matrix
$\ExpMat_{\B^*}(\Hm)$ is a parity--check matrix of $\ExpCode_\B(\CC)$
and, from Proposition~\ref{prop:dual_pu_sh}, removing (puncturing) the
last $m-\lambda$ columns of each block of this matrix provides a
parity--check matrix of $\ExpCode_{\BS}(\CCSS)$.
\fi

\subsection{Puncturing and shortening}

\begin{lem}
  \label{lem:pu-sh-commute}
  Let $\CC$ be an $[n,k]$ code over $\Fqm$. Let $\Lcal$ denote a
  subset of $\IInt{0}{n-1}$. Then the following equalities hold.
  \[ \ExpCode_{\B}(\pu{\CC}{\Lcal}) = \pu{\ExpCode_{\B}(\CC)}{\Lexp}, \]
  \[ \ExpCode_{\B}(\sh{\CC}{\Lcal}) = \sh{\ExpCode_{\B}(\CC)}{\Lexp}, \]
  where $\Lexp$ denotes the set of all columns generated from
  expanding columns in $\Lcal$, that is
\[\Lexp \eqdef \bigcup_{i \in \Lcal} \{i+j, 0 \leq j < m \}.\]
\end{lem}

\begin{proof}
  The result is straightforward for puncturing. The expansion
  operation is independent for each column, hence puncturing a column
  before expanding is equivalent to puncturing the corresponding block
  of $m$ columns. As for shortening, the shortening operation is the
  dual of puncturing operation, hence the result is a consequence of
  the next lemma.
\end{proof}

\subsection{Dual code}

\begin{lem}[\cite{W11}, Lemma~1]
  Let $\B$ be a basis and $\Bdual$ denote the dual basis. For all
  $\av, \bv \in \Fqm^n$, if $\av$ and $\bv$ are orthogonal, i.e.
  $\av \cdot \transpose{\bv} = 0,$ then $\ExpVect_\B(\av)$ and
  $\ExpVect_{\Bdual}(\bv)$ are orthogonal
\[ \ExpVect_{\B}(\av) \cdot \transposeb{\ExpVect_{\Bdual}(\bv)} = 0.\]
\end{lem}

\begin{cor}
  \label{cor:dual}
  Let $\CC$ be an $[n,k]$ code over $\Fqm$. Let $\Ball$ be a basis of $\Fqm$. Then the following equality holds.
  \[ \Dual(\ExpCode_\B(\CC)) = \ExpCode_{\Bdual}(\Dual(\CC)), \]
where $\Bdual$ denotes the dual basis of $\B$.
\end{cor}

\subsection{Changing the expansion basis}

\begin{lem}
  \label{lem:basis-change}
  Let $\CC$ be an $[n,k]$ code over $\Fqm$. Let $(\B_i)_i$ be $n$
  bases of $\Fqm$. Let $(\Qm_i) \in (\Fq^{m \times m})^n$
  denote $n$ invertible $m \times m$ matrices. The following equality
  holds.
  \[ \ExpCode_{(\B_i \cdot \transpose{(\Qm_i^{-1})})_i}(\CC) =
    \ExpCode_{\B_i}(\CC) \cdot 
    \begin{pmatrix}
      \Qm_0 & & \\
      & \ddots & \\
      & & \Qm_n \\
    \end{pmatrix}
    .\]
\end{lem}

\begin{proof}
  Let $\cv$ be a codeword of $\CC$. We only focus on the first
  entry of $\cv$ which we denote by $c$.
  We also denote the elements of $\B_0$ by $(b_0, \dots, b_{m-1})$
  and the entries of $\Qm_0$ by ${(q_{i,j})}_{i,j}$.
  The decomposition of $c$ in
  $\B_0$ is $c = c_{0}b_{0} + \cdots + c_{m}b_{m}$.
  Let $\mathcal{D} = (d_0, \ldots, d_{m-1})$ be the basis 
  $\B_0 \cdot \transposeb{\Qm_0^{-1}}$. For any $i \in \IInt{0}{m-1}$,
  we have $b_{i} = \sum_{j=0}^{m-1} d_j q_{i,j}$.
  Expressing $c$ in this new basis gives
  \[
    c = \sum_{i=0}^{m-1} c_{i} b_{i} =
    \sum_{i}\bigg(\sum_{j}d_j q_{i,j} \bigg) = \sum_{j}
    \bigg(\sum_{i} x_i q_{i,j}\bigg) d_j.
  \]
  Therefore, 
  \(\ExpVect_\B(c) \cdot \Qm_0 = \ExpVect_{\mathcal{D}}(c).\)
  This holds for any entry of any codeword $\cv \in \CC$.
\end{proof}

\subsection{Scalar multiplication in \Fqm}

\begin{lem}
  \label{lem:Fqm-scalar-multiplication}
  Let $\CC$ be an $[n,k]$ code over $\Fqm$. Let $(\B_i)_i$ be $n$
  basis of $\Fqm$. Let $\av = (a_0, \ldots, a_{n-1}) \in \Fqm^{n}$
  denote a vector of length $n$ over $\Fqm$. The following equality
  holds.
\begin{align*}
  \ExpCode_{(\B_i)_i}(\CC) & = \ExpCode_{(a_i \B_i)_i}( \{ (\cv \star \av) \,|\, \cv \in \CC \} ) \\
                           & = \ExpCode_{(a_i \B_i)_i}(\CC \star \av).
\end{align*}
\end{lem}


\section{The XGRS cryptosystem}\label{sec:SSRS_cryptosystem}

\subsection{The cryptosystem}

We describe here the cryptosystem which we call {\em XGRS}, presented
in \cite{KRW21} by Khathuria, Rosenthal and Weger. Next, we explain
why XGRS is a sub-instance of the SSRS scheme.

\subsubsection{Parameters}

The cryptosystem is publicly parametrised by:
\begin{itemize}
\item $q$ a prime power;
\item $m$ an integer; 
\item $\lambda$ such that $2 \leq \lambda < m$; 
\item $n,k$ such that $0 \leq k < n \leq q^m$ and $km > (m - \lambda)n$.
\end{itemize}

\begin{table}[h]
  \centering
  \begin{tabular}{|c|c|c|c|c|c|} \hline
    $q$ & $m$ & $\lambda$ & $n$ & $k$ & Public Key Size (kB) \\ \hline 
    13 & 3 & 2 & 1258 & 1031 & 579 \\ 
    7 & 4 & 2 & 1872 & 1666 & 844 \\ \hline
  \end{tabular}
  \medskip
  \caption{Parameters proposed for the XGRS scheme \cite{KRW21}}
  \label{tab:parameters}
\end{table}

\begin{remark}
  As suggested by the parameters in Table~\ref{tab:parameters}, $m$ is
  a small integer. The preprint version of the paper \cite{KRW19}
  proposed to use $m=2$ with a slightly modified key generation. The
  proposed parameters are now $m=3$ and $m=4$.
\end{remark}

\subsubsection{Key Generation}

\begin{itemize}
\item Generate uniformly random vectors
  $(\xv, \yv) \in \Fqm^n \times (\Fqm^\times)^n$ such that $\xv$ has
  distinct entries. Denote $\CC = \GRS{k}{\xv}{\yv}$ and let $\Hsec$
  be a parity-check matrix of $\CC$.
\item Choose $\gamma$, a primitive element of $\Fqm/\Fq$, \ie a
  generator of the field extension.  We consider the basis $\Bgammaall$
  of $\Fqm$.  
\item Set
  $\Hm \eqdef \ExpMat_{\Bgamma^{\ast}}(\Hsec) \in \Fq^{m(n-k) \times mn}$
  a parity-check matrix of $\ExpCode_{\Bgamma}(\CC)$.
\item For any $i \in \IInt{0}{n-1}$, choose $\L_i$ a random subset of
  $\IInt{(i-1)m}{im-1}$ of size $|\L_i|=m-\lambda$. Set
  $\L = \cup_i \L_i$.
\item Set $\HmL \eqdef \pu{\Hm}{\L} \in \Fq^{m(n-k) \times \lambda n}$.
\item For any $i \in \IInt{0}{n-1}$, choose $\Qm_i$ a random
  $\lambda \times \lambda$ invertible matrix.
  Denote by $\Qm$ the block-diagonal matrix having
  $\Qm_0, \ldots, \Qm_{n-1}$ as diagonal blocks.
\item Denote by $\Sm$ the invertible matrix of $\Fq$ such that
  $\Sm \,\cdot\, \HmL \,\cdot\, \Qm$ is in systematic form.
\item Set $\Hpub \eqdef \Sm \,\cdot\, \HmL \,\cdot\, \Qm.$
\item The public key is $\Hpub$, the private key is
  $(\xv, \yv, \Qm, \L, \gamma)$.
\end{itemize}

\begin{remark}
  Compared to the cryptosystem presented in \cite{KRW21}, we omitted
  the block permutation. Indeed, applying a block permutation after
  expanding is equivalent to applying the permutation before the
  expansion and then expanding. As we start with a GRS code chosen
  uniformly at random, applying a permutation on the columns does not
  change the probability distribution of the public keys.
\end{remark}

\subsubsection{Encryption}

Recall that $t \eqdef \lfloor \frac{n-k}{2} \rfloor$ the
error--correcting capacity of a GRS code of length $n$ and dimension
$k$.  The message is encoded as a vector $\yv \in \Fq^{\lambda n}$
whose support is included in $t$ blocks of length $\lambda$, \ie there
exist positions $i_0, \ldots, i_{t-1} \in \IInt{0}{n-1}$, such that
\[\Sup(\yv) \subseteq \bigcup_{0 \leq \ell \leq t-1}
\IInt{\lambda(i_\ell-1)}{\lambda i_\ell - 1}.\]
The ciphertext is then defined as
$\transpose{\cv} = \Hpub \,\cdot\, \transpose{\yv}$.

\subsubsection{Decryption}

In order to decrypt the ciphertext, a user knowing the private key
should:
\begin{itemize}
\item generate $\Hsec$ from $\xv$ and $\yv$.
\item compute $\cv' = \cv \,\cdot\, \transpose{\Sm^{-1}}$;
\item compute $\cv'' = \SqueezeVect_{\B_\gamma}(\cv')$;
\item find $\yv'' \in \Fqm^n$ of weight $|\yv''| \leq t$ such that
  $\transpose{\cv''} = \Hsec \transpose{\yv''}$ (\ie decode in $\CC$);
\item compute $\yv' = \pu{\ExpVect_{\B_\gamma}(\yv'')}{\Lcal}$;
\item finally recover $\yv = \yv' \,\cdot\, \transposeb{\Qm^{-1}}$.
\end{itemize}

\subsection{Relation with the SSRS scheme}
\label{subsub:relation_with_gen_scheme}
\label{subsec:practical_construction}

To conclude this section, we show that the XGRS scheme is a
sub-instance of the SSRS scheme presented in
Section~\ref{sec:generic-scheme}.

\begin{proposition}
  The XGRS scheme with secret key $(\xv, \yv, \Qm, \L, \gamma)$ is
  equivalent to the SSRS scheme with secret key
  $(\xv, \SS_0, \dots, \SS_{n-1})$ where the subspaces $\SS_i$ are
  defined as follows.
  \begin{itemize}
   \item Let $\B_i^{(0)} \eqdef \pu{\Bgamma}{\L_i} \in \Fqm^\lambda$ where
    $\L_i \eqdef \{j-mi, \forall j \in \Lcal \cap \IInt{im}{(i+1)m-1}\}$.
  \item Set
    $\B_i^{(1)} \eqdef y_i^{-1} \B_i^{(0)} \cdot
    \transpose{(\Qm_i^{-1})}$.
  \item $\SS_i$ is the subspace of $\Fqm$ spanned by the elements of
    $\B_i^{(1)}$.
  \end{itemize}
\end{proposition}

\begin{proof}
  Let $\Cpub$ denote the public code of an instance of the XGRS scheme
  with private key $(\xv, \yv, \Qm, \L, \gamma)$, \ie $\Cpub$ is the
  code over $\Fq$ that admits the public key $\Hpub$ as parity-check
  matrix. We have
  \begin{align*}
    \Cpub & =  \Dualb{\Fqspan{\Hpub}} \\
          & = \Dualb{\Fqspan{  \HmL \cdot \Qm}}.
  \end{align*}
  
  Let us define $\Qm^{(1)} \eqdef \transposeb{\Qm^{-1}}$. This is
  still a block-diagonal matrix composed of $n$ blocks of size
  $\lambda \times \lambda$. We can rewrite this
  \[\Cpub =  \Dualb{\Fqspan{ \HmL}} \cdot \Qm^{(1)} .\]
  We can replace $\HmL$ by its definition:
  $\pu{\ExpMat_{\Bgamma^{\ast}}(\Hsec)}{\L}$. Next, we can swap the
  $\Dual$ and $\mathbf{Punct}$ operators according to
  Proposition~\ref{prop:dual_pu_sh}:
  \[
    \Cpub = \sh{\Dual
      \left(\ExpCode_{\Bgamma^{\ast}}\left(\Fqspan{\Hsec}\right)\right)}{\L}
    \cdot \Qm^{(1)}.
  \]
  We can then swap the $\Dual$ and $\ExpCode$ operators according to
  Corollary~\ref{cor:dual}.
  \[
    \Cpub =
    \sh{\ExpCode_{\Bgamma}\left(\Dualb{\Fqspan{\Hsec}}\right)}{\L}
    \cdot \Qm^{(1)} .
  \]

  Let $\Gsec$ be a generator matrix of the secret code
  $\Dualb{\Fqspan{\Hsec}}$, \ie a generator matrix of the code
  $\GRS{k}{\xv}{\yv}$. We have
  \[
    \Cpub = \sh{\ExpCode_{\Bgamma}\left(\Fqspan{ \Gsec}\right)}{\L}
    \cdot \Qm^{(1)}.
  \]

  Let us denote $\Qm^{(2)}$ the block-diagonal matrix obtained by
  replacing each $\lambda \times \lambda$ matrix of $\Qm^{(1)}$ by the
  $m \times m$ matrix obtained by inserting ``an identity row/column''
  at the positions corresponding to $\Lcal$. For instance, if
  $m = 3, \lambda = 2$ and the first element of $\L$ equals $1$, which
  means that the column $1$ is shortened, we add a column and a row in
  the middle of $\Qm^{(1)}_0$, \ie
  \[ 
    \text{if } 
    \qquad
    \Qm^{(1)}_0 = \left( 
      \begin{array}{cc}
        q_{00} & q_{01} \\
        q_{10} & q_{11}
      \end{array}
    \right) 
    , 
    \qquad
    \text{then }
    \qquad
    \Qm^{(2)}_0 = \left( 
      \begin{array}{ccc}
        q_{00} & 0 & q_{01} \\
        0 & 1 & 0 \\
        q_{10} & 0 & q_{11}
      \end{array}
    \right) 
    .
  \]
  Hence, we can write
  \[
    \Cpub = \sh{\Fqspan{ \ExpMat_{\Bgamma}(\Gsec) \cdot
        \Qm^{(2)}}}{\L}.
  \]
  We define $\Qm^{(3)}_i$ as the matrix obtained from $\Qm^{(2)}_i$ by
  permuting the columns so that the inserted columns are the
  $m-\lambda$ rightmost ones.  For instance in the previous example,
  we would have
  \[
    \Qm^{(3)}_0 = \left( 
      \begin{array}{ccc}
        q_{00} & q_{01} & 0  \\
        0 & 0 &  1\\
        q_{10} & q_{11} & 0 
      \end{array}
    \right) 
    .
  \]
  Therefore, $\Qm^{(3)} = \Qm^{(2)}\Pm$ where $\Pm$ is a
  block--diagonal matrix whose diagonal blocks are $m \times m$
  permutations matrices.
  Then, we replace $\Lcal$ by the set
  $\JLM = \{mi+j ~|~ 0 \leq i < n,\ \lambda \leq j < m\}$.
  Hence, we get
  \[ \Cpub = \sh{\Fqspan{ \ExpMat_{\Bgamma}(\Gsec) \cdot
        \Qm^{(3)}}}{\JLM}.
  \]
  We can apply the basis change explained in
  Lemma~\ref{lem:basis-change}.
  \[
    \Cpub = \sh{\Fqspan{ \ExpMat_{{(\B'_i)}_i}(\Gsec)}}{\JLM},
  \]
  where $\B'_i \eqdef \Bgamma \cdot \transposeb{(\Qm^{(3)}_i)^{-1}}$
  for all $i \in \IInt{0}{n-1}$.
  Finally, we apply Proposition~\ref{prop:scalar} and
  Lemma~\ref{lem:Fqm-scalar-multiplication} to replace the code
  $\GRS{k}{\xv}{\yv}$ by $\RS{k}{\xv}$. Hence, 
  \[
    \Cpub = \sh{\Fqspan{ \ExpMat_{{(\B_i)}_i}(\Gsec')}}{\JLM},
  \]
  where $\Gsec'$ is a generator matrix of $\RS{k}{\xv}$ and
  $\B_i \eqdef y_i^{-1} \B'_i$ for all $i \in \IInt{0}{n-1}$.  In
  other words, $\Cpub = \RS{k}{\xv}_{|(\SS_0, \ldots, \SS_{n-1})}$,
  where $\SS_i$ is the subspace spanned by the $\lambda$ first elements
  of $\B_i$. This is indeed an instance of the SSRS cryptosystem.
\end{proof}


\section{Twisted-square code and distinguisher}
\label{sec:distinguisher}

In this section, we first recall how the star product of codes may
distinguish some codes such as Reed--Solomon codes from random
ones. It turns out that the raw use of the star product is not
sufficient to distinguish SSRS codes from random codes. Therefore, in
the second part of this section, we introduce a new notion called {\em
  twisted star product}.  We first explore the case $m=3$, $\lambda=2$
to give some insight on the interest of defining this twisted star
product of two subspace subcodes. Then, we define this notion for
other parameters. Note that a similar notion appears in
Randriambololona's work on the existence of asymptotically good binary
codes with asymptotically good squares \cite{R13b}. In the third part
of the section, we focus on the dimension of the twisted square codes
and how it can be used as a distinguisher.

We refer to Figure~\ref{fig:recap} for a diagram summarising the
relations between the various statements to follow.

\subsection{Square code distinguisher}
The notion of star product has been recalled in
Section~\ref{subsec:starprod}. We recall the following result on the
generic behaviour of random codes with respect to this operation.

\begin{thm}({\cite[Theorem 2.3]{CCMZ15}}, informal)\label{thm:CCMZ}
For a linear code $\code R$ chosen at random over $\Fq$ of dimension
$k$ and length $n$, the dimension of $\code R^{\star 2}$ is typically
$\min (n, {k+1 \choose 2})$.
\end{thm}

Theorem~\ref{thm:CCMZ} provides a distinguisher between random codes
and algebraically structured codes such as generalised Reed--Solomon
codes and their low--codimensional subcodes~\cite{W10, CGGOT14},
Reed--Muller codes~\cite{CB14}, polar codes~\cite{BCDOT16} some Goppa
codes \cite{COT14a, COT17}, high--rate alternant codes \cite{FGOPT11}
or algebraic geometry codes~\cite{CMP14,CMP17}.  For instance, in the
case of GRS codes, we have the following result.
\begin{proposition}\label{prop:sq_RS} Let $n, k, \xv, \yv$ be as in
Definition~\ref{def:GRS}. Then,
  \[ \sqb{\GRS{k}{\xv}{\yv}} = \GRS{2k-1}{\xv}{\yv \star \yv}.
  \] In particular, if $k \leq n/2$, then
 \[ \dim \sqb{\GRS{k}{\xv}{\yv}} = 2k-1.
 \]
\end{proposition}

Thus, compared to random codes whose square have dimension quadratic
in the dimension of the code, the square of a GRS code has a dimension
which is linear in that of the original code. This criterion allows to
distinguish GRS codes of appropriate dimension from random codes.
A rich literature of cryptanalysis of code--based encryption primitives
involves this operation.

In an initial version of the XGRS cryptosystem, submitted on ArXiv
\cite{KRW19}, such a distinguisher could be applied to the
cryptosystem and lead to an attack. The authors changed the
cryptosystem in order to avoid such attacks. The new parameters are
out of reach of a distinguisher based on the square code operation.
This motivates the introduction of an alternative product.

\subsection{The twisted square product}
\label{ss:twisted_star_prod}
In the remainder of this section, for the sake of simplicity, are
stated using the same subspace and expansion basis for all blocks but
they can be straightforwardly generalised to the case of various
subspaces and expansion bases.

\subsubsection{Motivation: the case $\lambda=2, m=3$}
In this section, we consider the case $\lambda=2, m=3$.  Let us
introduce a definition that is needed in the sequel.

\begin{definition}
  Let $\SS \subseteq \Fqm$ be an $\Fq$--vector space, we define
  the square subspace
  \[
  \SS^2 \eqdef \Fqspan{ab ~|~ a, b \in \SS}.
  \]
\end{definition}

\begin{lem} \label{lem:product}
  Let $\SS$ be a subspace of $\Fqm$ of dimension 2. Let
  $\BS = (\gamma_0, \gamma_1)$ be a basis of $\SS$.
  Let $a,b \in \SS$ such that 
  \[
  \ExpVect_\BS((a)) = (a_0,a_1) \quad {\rm and} \quad
  \ExpVect_\BS((b)) = (b_0,b_1). 
  \]
  Then,
  \[
  \ExpVect_{\BSp}((ab)) = (a_0b_0,a_0b_1+a_1b_0,a_1b_1),
  \]
  where $\BSp = (\gamma_0^2,\gamma_0\gamma_1,\gamma_1^2)$.
\end{lem}

\begin{remark}
  Note that when $m=3$ and $\dim \SS=2$, we have $\SS^2 = \F_{q^3}$.
  Indeed, let $(\gamma_0, \gamma_1)$ be a basis of $\SS$, if
  $\gamma_{0}^2, \gamma_{0}\gamma_{1}$ and $\gamma_{1}^2$ were not
  $\Fq$-independent, denoting $\zeta \eqdef \gamma_1 / \gamma_0$, then
  $1, \zeta$ and $\zeta^2$ would not be $\Fq$-independent
  either. Hence $\zeta$ would have degree $\leq 2$ over $\Fq$. But by
  definition $\zeta \not\in \Fq$.
\end{remark}

Consider the SSRS scheme with $m=3, \lambda=2$. The public key is a
generator matrix $\Gm$ of
$\ExpCode_{(\B_{\SS_0}, \dots,
  \B_{\SS_{n-1}})}\left(\SScode{\CC}{(\SS_0, \dots,
    \SS_{n-1})}\right)$.
From an attacker's point of view, the spaces $\SS_0, \dots, \SS_{n-1}$
and their bases $\B_{\SS_0}, \dots, \B_{\SS_{n-1}}$ are unknown.  But,
we have access to the entries of $\Gm$, in particular we have access
to the coefficients $a_0, a_1$ (resp. $b_0, b_1$) of the decomposition
in the basis $\B_i$ of the $i$--th entry of some codewords of
$\CC$. Hence, the coefficients of the product $ab$ in the basis
$(\gamma_0^2, \gamma_0 \gamma_1, \gamma_1^2)$ of $\SS_i^2$ can be
computed {\bf without knowing neither $\CC$ nor the basis
  $\B_i$}. This motivates the following definition.

\begin{definition}[Twisted product]
  Let $\av$ and $\bv$ in $\F_q^{2n}$ whose components are denoted
  \[ \av =
  (a_0^{(0)},a_0^{(1)},a_1^{(0)},a_1^{(1)},\ldots,a_{n-1}^{(0)},a_{n-1}^{(1)});\]
  \[ \bv =
  (b_0^{(0)},b_0^{(1)},b_1^{(0)},b_1^{(1)},\ldots,b_{n-1}^{(0)},b_{n-1}^{(1)}).\]
  We define the \emph{twisted product} of $\av$ and $\bv$ as
  \[
  \av \,\tstar\, \bv \eqdef (a_i^{(0)}b_i^{(0)}, a_i^{(0)}b_i^{(1)}+a_i^{(1)}b_i^{(0)} , a_i^{(1)}b_i^{(1)})_{0 \leq i \leq n-1} \in \Fq^{3n}.
  \]
  This definition extends to the product of codes, where the {\em twisted
    product} of two codes $\code A$ and $\code B \subseteq \F_q^{2n}$ is
  defined as
  \[
  \code A \tstar \code B \eqdef \Fqspan{ \av \,\tstar\, \bv ~|~ \av \in \code A, \ \bv \in \code B }.
  \]
  In particular, $\tsq{\code A}$ denotes the \emph{twisted square code}
  of a code $\code A$: $\tsq{\code A}\eqdef \code A \tstar \code A$.
\end{definition}

With this definition, we can rewrite Lemma~\ref{lem:product} for
vectors in the following way.

\begin{lem} \label{lem:twisted_product}
  Let $\SS$ be a subspace of $\Fqm$ of dimension 2. Let
  $\BS = (\gamma_0, \gamma_1)$ be a basis of $\SS$. Let
  $\av, \bv \in \Fqm^n$ such that all their entries lie in
  $\SS$. Then,
  \begin{equation*}
    \ExpVect_\BS(\av) \,\tstar\, \ExpVect_\BS(\bv) = \ExpVect_{\BSp}(\av \star \bv),
  \end{equation*}
  where $\BSp = (\gamma_0^2,\gamma_0\gamma_1,\gamma_1^2)$.  
\end{lem}

Extending this result to codes, we obtain the following theorem.

\begin{thm}\label{thm:twisted_sq_exp_code}
  Let $\CC$ be an $[n,k]$-code over $\Fqm$ and $\SS$ a subspace of
  $\Fqm$ of dimension $\lambda$. Let $\BS = (\gamma_{0},\gamma_{1})$ be
  an $\Fq$--basis of $\SS$. Then,
  \begin{equation}\label{eq:twisted_sq_exp_code_simple}
    \tsqb{\ExpCode_{\BS} (\SScode{\CC}{\SS})} =
    \ExpCode_{\BSp}\left(\left(\CCSS\right)^{\star 2}_{\Fq}\right),
  \end{equation}
  where $\BSp = (\gamma_{0}^2, \gamma_{0}\gamma_{1}, \gamma_{1}^2)$.
  This results generalises straightforwardly to an expansion over
  various subspaces.
\end{thm}

\begin{proof}
  This is a direct consequence of Lemma~\ref{lem:twisted_product} by
  definition of the $\ExpCode$ operator and by $\Fq$-linearity of
  $\ExpVect$.
\end{proof}

Note that in our case, because $\SS^2 = \Fqm$, the basis $\BSp$ in
\eqref{eq:twisted_sq_exp_code_simple} is a full basis of $\Fqm$.
Theorem~\ref{thm:twisted_sq_exp_code} is summarised by the following
diagram.

\begin{center}
\begin{tikzcd}
  \CC
  \ar[r, "(\,\cdot\,)_{|\SS}"]
  \ar[d, "\mathbf{Exp}_{\B}"]
  & [5em]
  \CCSS
  \ar[r, "\sqb{\,\cdot\,}_{\Fq}"]
  \ar[d, "\mathbf{Exp}_{\BS}"]
  & [4em]
  \sqb{\CCSS}_{\Fq}
  \ar[d, "\mathbf{Exp}_{\B_{\SS^2}}"]
  \\ [3em]
  \ExpCode_{\B}(\CC)
  \ar[r, "\sh{\,\cdot\,}{\JLM}"]
  &
  \ExpCode_{\BS}(\CC_{|\SS})
  \ar[r, "\tsqb{\,\cdot\,}"]
  &
  \ExpCode_{\B_{\SS^2}}\left(\sqb{\CCSS}_{\Fq}\right)
  \\
\end{tikzcd}
\end{center}


\subsubsection{General definition of the twisted square code}
\begin{remark}
  This section is a generalisation of the previous definitions and
  results. A reader only interested in the practical aspects of the
  attack can skip directly to
  Section~\ref{sec:twisted_sq_distinguisher}.
\end{remark}

For arbitrary $\lambda \geq 2$, we have the following definition.

\begin{definition}[Twisted square product, general case]
  Let $\av$ and $\bv$ in $\F_q^{\lambda n}$ whose components are denoted
  \begin{align*}
    \av &= (a_{0,0}, \dots, a_{0,\lambda - 1}, a_{1,0}, \dots,
          a_{1,\lambda - 1},\ldots,a_{n-1,0}, \dots, a_{n-1,\lambda
          - 1}),\\
    \bv &= (b_{0,0}, \dots, b_{0,\lambda - 1},b_{1,0}, \dots,
          b_{1,\lambda - 1},\ldots,b_{n-1, 0}, \dots, b_{n-1, \lambda - 1}).
  \end{align*}
  We define the \emph{twisted product}
  $\av \, \tstar\, \bv \in \Fq^{{\lambda + 1 \choose 2}n}$ of $\av$
  and $\bv$ such that for any $i \in \IInt{0}{n-1}$ and for $r,s$ such
  that $0 \leq r \leq s \leq \lambda - 1$,
  \[{(\av \,\tstar\, \bv)}_{i
    {\lambda + 1 \choose 2} + {s+1 \choose 2} + r} \eqdef \left\{
    \begin{array}{ccc}
      a_{i,r}b_{i,s} + a_{i,s}b_{i,r}& {\rm if} & r < s\\
      a_{i,r}b_{i,r} & {\rm if} & r = s.
    \end{array}
  \right.
  \]
  This definition extends to the product of codes, where the {\em twisted
    product} of two codes $\code A$ and $\code B \subseteq \F_q^{\lambda n}$ is
  defined as
  \[
  \code A \tstar \code B \eqdef \Fqspan{
    \av \,\tstar\, \bv ~|~ \av \in \code A, \ \bv \in \code B
  }.
  \]
  In particular, $\tsq{\code A}$ denotes the \emph{twisted square code}
  of a code $\code A$: $\tsq{\code A}\eqdef \code A \tstar \code A$.
\end{definition}

We are interested in the case where $\SS^2 = \Fqm$, because the goal
is to reconstruct a fully expanded code. For a random subspace
$\SS \subseteq \Fqm$ of dimension $\lambda$, its dimension is
typically $\min\left\{{\lambda+1 \choose 2},m\right\}$. The case
$m=3, \lambda=2$ is a special case where ${\lambda+1 \choose 2}=m$.
Hence, for other parameters $m,\lambda$ such that
${\lambda + 1 \choose 2} = m$, Theorem~\ref{thm:twisted_sq_exp_code}
generalises straightforwardly. But when ${\lambda + 1 \choose 2} > m$,
the twisted square code does not correspond to an expanded code. It is
as if the code were expanded over a generating family of $\Fqm$ which
is not a basis: the vectors are too long. A way to circumvent this and
to obtain a result similar to Theorem~\ref{thm:twisted_sq_exp_code} is
to shorten the twisted square code to cancel the useless columns and
obtain an expansion over a basis. This yields the following results.

\begin{lem} \label{lem:twisted_product_gen} Let $\SS$ be a subspace of
  $\Fqm$ of dimension $\lambda$ such that $\SS^2 = \Fqm$. Let
  $\BS = (\gamma_0, \ldots, \gamma_{\lambda-1})$ be an $\Fq$--basis of $\SS$.
  Let $\B_{\SS^2}$ denote the first $m$ elements of
  $(\gamma_0^2, \gamma_0 \gamma_1, \dots, \gamma_0\gamma_i,$
  $\gamma_1
  \gamma_i, \dots,$ $\gamma_i^2,\dots, \gamma_{\lambda - 1}^2)$.
  Let $\av, \bv \in \Fqm^n$ whose entries all lie in $\SS$.  Denote
  $\cv$ the vector of length ${\lambda + 1 \choose 2}n$ over $\Fq$
  defined as
  \[\cv \eqdef \ExpVect_\BS(\av) \,\tstar\, \ExpVect_\BS(\bv).\]
  Let $\KLM$ denote the set
  \begin{equation}\label{eq:J0}
    \KLM \eqdef \left\{ {\textstyle{ \lambda+1 \choose 2}}i+j, i \in \IInt{0}{n-1}, j \in \IInt{m}{\textstyle{ \lambda+1 \choose 2}-1} \right\}.
  \end{equation}
  If $\BSp$ is a basis of $\Fqm$ and for any $i\in \KLM$, the
    $i$--th entry of $\cv$ is zero, then,
  \begin{equation}
    \label{eq:product}
    \pu{\cv}{\KLM} = \ExpVect_{\BSp}(\av \star \bv).
  \end{equation}
\end{lem}

\begin{proof}
  Let $\cv$ be defined as in the statement. We want to prove that 
  \[\SqueezeVect_\BSp(\pu{\cv}{\KLM}) = \av \star \bv.\]
  This is equivalent to Equation~\eqref{eq:product} because $\BSp$ is
  a basis of $\Fqm$.  Without loss of generality, we only need to
  focus on the block corresponding to the first entry in $\Fqm$.
 
  Let $(a_0, \ldots, a_{\lambda-1})$ and
  $(b_0, \ldots, b_{\lambda-1})$ denote the decomposition of the first
  entries of $\av$ (\emph{resp.} $\bv$) over $\BS$.
  The first entry of $\av \star \bv$ is
  \[ \left(\sum_i a_i \gamma_i \right)\left(\sum_j a_j \gamma_j\right)
  = \sum_{0 \leq i \leq j < \lambda} c_{i,j} \gamma_i \gamma_j,\]
  where the coefficients $c_{i,j}$ match exactly the definition of the 
  twisted square product, hence correspond to the entries of $\cv$.

  Let $\Bfull$ denote the family
  $(\gamma_0^2, \gamma_0 \gamma_1, \dots, \gamma_0\gamma_i, \gamma_1
  \gamma_i, \dots, \gamma_i^2,\dots, \gamma_{\lambda - 1}^2)$.
  The last entries of each block of $\cv$ are equal to zero. This
  corresponds exactly to the elements of $\Bfull$ that are not in
  $\BSp$. We therefore have
  \[ \SqueezeVect_\BSp(\pu{\cv}{\KLM}) = \SqueezeVect_\Bfull(\cv) =
  \av \star \bv.\]
\end{proof}

This leads to the following main statement, which is the
generalisation of Theorem~\ref{thm:twisted_sq_exp_code}.

\begin{thm}
  \label{thm:generalise_connect}
  Let $\CC$ be an $[n,k]$ code over $\Fqm$ and $\SS$ a subspace of
  $\Fqm$ of dimension $\lambda$ such that $\SS^2 = \Fqm$. Let
  $\BS = (\gamma_{0},\ldots, \gamma_{\lambda-1})$ be an $\Fq$--basis
  of $\SS$. Then,
  \begin{equation}\label{eq:twisted_sq_exp_code}
    \sh{\tsqb{\ExpCode_{\BS} (\SScode{\CC}{\SS})}}{\KLM} \subseteq
    \ExpCode_{\BSp}\left(\sqb{\CCSS}_{\Fq}\right),
  \end{equation}
  where $\BSp$ and $\KLM$ are defined as in
  Lemma~\ref{lem:twisted_product_gen}, provided $\BSp$ is a basis of
  $\Fqm$. This result generalises straightforwardly to an expansion
  over various subspaces and bases.
\end{thm}

\begin{proof}
  We intend to apply Lemma~\ref{lem:twisted_product_gen}. This lemma
  has two conditions.
  The first condition is met by shortening the left-hand term. Indeed,
  the effect of shortening is that we keep only the words whose
  entries indexed by $\KLM$ are all equal to zero.
  The second condition ($\BSp$ being a basis) is a hypothesis.
  
  Compared to Lemma~\ref{lem:twisted_product_gen} and its proof, one
  should be careful that in general $\sh{\tsq{\AC}}{\KLM}$ (where
  $\AC$ denotes $\ExpCode_{\BS} (\SScode{\CC}{\SS})$) is not spanned
  by words of the form $\pu{\av \tstar \bv}{\KLM}$ with
  $\av, \bv \in \ExpCode_{\B}(\AC)$ but by linear combinations, \ie
  words of the form
  \[
    \pu{\av_0 \tstar \bv_0 + \cdots + \av_s \tstar \bv_s}{\KLM},\quad
    \mathrm{for}\quad \av_0, \dots, \av_s, \bv_0, \dots, \bv_s \in \AC.
  \]
  Therefore, one needs to apply the very same reasoning as that of the
  proof of Lemma~\ref{lem:twisted_product_gen} replacing
  $\av \tstar \bv$ by a sum of such vectors. This is not a problem and
  the proof generalises straightforwardly, since all the involved
  operators are linear.

  Finally, because of the shortening operation, we only obtain an
  inclusion and not an equality.
\end{proof}

\begin{remark}
  In the special case ${\lambda+1 \choose 2} = m$, $\KLM = \emptyset$,
  therefore the shortening is useless and the inclusion in
  \eqref{eq:twisted_sq_exp_code} is an equality.
\end{remark}

\begin{remark}
  In the sequel, we see that under a reasonable
  conjecture and some condition, the inclusion in
  \eqref{eq:twisted_sq_exp_code} is an equality.
\end{remark}

\subsection{Dimension of the twisted square of subspace subcodes}
\label{sec:twisted_sq_distinguisher}

\subsubsection{Typical dimension of the twisted square of a
  random subspace subcode}

Similarly to Theorem~\ref{thm:CCMZ} on squares of random codes, we
expect that twisted squares of random codes typically have the largest
possible dimension. For this reason, we state the following conjecture
which is confirmed by our experimental observations using the computer
algebra software {\em Sage}.

\begin{conjec}\label{conj:twisted_sq}
  For any positive integer $k$ such that $2k \leq n$, any $\Fq$--subspace
  $\SS \subseteq \Fqm^n$ of dimension $\lambda \geq 2$ such that
  $\SS^2 = \Fqm$ and any $\Fq$--basis $\BS$ of $\SS$, let $\RC$
  denote an $[n,k]$ code chosen uniformly at random, then
  \begin{equation*}
    \Prob{\dim_{\Fq} \tsq{\left(\ExpCode_{\BS}{(\RC_{|\SS})}\right)} =
      \min\left\{ {\lambda + 1 \choose 2}n,
        {km - n(m-\lambda) + 1 \choose 2} \right\}}
    \underset{k \to \infty}{\longrightarrow} 1.
  \end{equation*}
\end{conjec}

It is worth noting that in general ${\lambda + 1 \choose 2} > m$.  In
such a case, as already mentioned before stating
Lemma~\ref{lem:twisted_product_gen}, the code
$\tsqb{\ExpCode_{\BS}{\RC_{|\SS}}}$ represents something {\bf which is
  not an expansion of a code with respect to a basis of $\Fqm$} but
rather a kind of expansion with respect to a family of generators of
the set $\SS^2$. This family is denoted $\Bfull$ in the proof of
Lemma~\ref{lem:twisted_product_gen}. This set spans $\SS^2$ but is not
linearly independent in general. Hence, given a vector with entries in
$\Fqm^n$, the decomposition with respect to this family of generators
is not unique. For this reason, it is difficult to identify the
twisted square code with the expansion of another code.

To ensure the unique decomposition, the key idea is to proceed as in
the statement of Theorem~\ref{thm:generalise_connect} and to {\bf
  shorten} the twisted square code on the positions of the set $\KLM$
introduced in~(\ref{eq:J0}), \ie shortening the last
${\lambda + 1 \choose 2} - m$ positions of each block of length
${\lambda + 1 \choose 2}$. According to
Theorem~\ref{thm:generalise_connect}, a codeword in
$\sh{\tsq{\ExpCode_{\BS}(\SScode{\RC}{\SS})}}{\KLM}$ is the expansion
of a codeword of $\sq{\RC}$ in a given basis of $\Fqm$.  The latter
property is in general not satisfied by codewords of
$\tsq{\ExpCode_{\BS}(\SScode{\RC}{\SS})}$.  Therefore, this shortened
code turns out to be a more relevant object of study and its dimension
is of particular interest in the sequel.  This dimension is the
purpose of the following statement.
\begin{cor}\label{cor:dim_twisted_square_shortened_random}
  Let $\RC$ be a uniformly random $[n,k]$ code over $\Fqm$ and
  $\SS \subseteq \Fqm$ be a subspace such that $\SS^2 = \Fqm$. Denote
  by $\KLM$ the set introduced in~(\ref{eq:J0}).  Then, under
  Conjecture~\ref{conj:twisted_sq}, we typically have
  \begin{align}\label{eq:typical_square}
    \begin{split}
    \dim_{\Fq} & \sh{\tsq{\ExpCode_{\BS}(\SScode{\RC}{\SS})}}{\KLM} \geq \\
    &\quad \qquad \qquad \min \left\{mn, {km - n(m-\lambda) + 1 \choose 2} -
      n\left({\lambda + 1 \choose 2} - m \right) \right\}.
    \end{split}
  \end{align}
\end{cor}
  
\subsubsection{Typical dimension of the twisted square of a subspace
    subcode of a RS code}

On the other hand, subspace subcodes of Reed--Solomon codes have a
different behaviour. Indeed, Theorem~\ref{thm:generalise_connect} 
yields the following result.

\begin{cor}\label{cor:dim_sq_SSRS}
  Given a GRS code $\CC = \GRS{k}{\xv}{\yv}$ and an
  $\Fq$--subspace  $\SS \subseteq \Fqm$ of dimension $\lambda < m$
  such that $\SS^2 = \Fqm$.
  Denote by $\KLM$ the set introduced in~(\ref{eq:J0}). 
  Then,
\begin{equation}\label{eq:dim_square_SSRS}
    \dim_{\Fq} \left(
      \tsq{
        \sh{\ExpCode_{\BS}(
          \SScode{\CC}{\SS}
          )}
        {\KLM}}
    \right) \leq \min \{mn,m(2k-1)\}.
\end{equation}  
\end{cor}

\subsubsection{The distinguisher}\label{subsec:distinguisher}
Putting the previous statements together, the twisted product provides
a distinguisher between expanded subspace subcodes of GRS codes and
expanded subspace subcodes of random
codes.

\begin{thm}\label{thm:distinguisher}
  Let $k$ be a positive integer,
  $\SS \subseteq \Fqm^n$ of dimension $\lambda \geq 2$ an
  $\Fq$--subspace such that $\SS^2 = \Fqm$, $\BS$ an $\Fq$--basis. Let
  $\code D$ be defined as $\ExpCode_{\BS}(\SScode{\CC}{\SS})$, where
  $\CC$ is either a random $[n, k]$ code over $\Fqm$ or an $[n, k]$
  GRS code over $\Fqm$. Suppose also that
  \begin{equation}\label{eq:condition_distinguisher}
  m(2k-1) < \min \left\{mn, {km - n(m-\lambda) + 1 \choose 2} -
    n\left({\lambda + 1 \choose 2} - m \right) \right\}.
  \end{equation}
  Then, assuming Conjecture~\ref{conj:twisted_sq}, the computation of
  $ \dim_{\Fq}\sh{\tsq{\code D}}{\KLM} $ provides a polynomial-time
  algorithm which decides whether $\CC$ is an RS code or a random code
  and succeeds with high probability. This extends straightforwardly
  to the case of multiple spaces and bases.
\end{thm}

\begin{remark}\label{rem:rate}
  Condition (\ref{eq:condition_distinguisher}) entails in particular
  $2k \leq n$, which is a necessary condition for the distinguisher to
  succeed. Indeed, if $2k > n$, the square code of the GRS code spans
  the whole space $\Fqm^n$. Hence it cannot be distinguished from a
  random code.  When this condition is not met, it is sometimes
  possible to shorten the code so that the shortened code meets this
  condition. This is addressed in
  Section~\ref{subsec:broadening_by_shortening}.
\end{remark}

\subsubsection{Experimental results}
\label{ex:num_app}

Using the computer algebra software \emph{Sage}, we tested the
behaviour of the dimension of the twisted square (shortened at $\KLM$)
of subspace subcodes either of random codes or of RS codes. For each
parameter set (see Table~\ref{tab:param_tests}), we ran more than 100
tests and none of them yielded dimensions of the twisted square that
was different from the bounds given either by
Conjecture~\ref{conj:twisted_sq} or by
Corollary~\ref{cor:dim_sq_SSRS}.

\begin{table}[h] \centering
  \begin{tabular}{|c|c|c|c|c|c|c|c|} \hline \multirow{3}{*}{Parent code}&
\multirow{3}{*}{$q$} & \multirow{3}{*}{$m$} &
\multirow{3}{*}{$\lambda$} & \multirow{3}{*}{$n$} &
\multirow{3}{*}{$k$} & Bounds on & Actual\\ & & & & & &
the dimension&Dimension\\ & & & & & & of $\sh{\tsq{\CC}}{\KLM}$&of
$\sh{\tsq{\CC}}{\KLM}$\\ \hline Random & 7 & 3 & 2 & 120 & 55 & $\geq 360$
& 360 \\ \hline RS & 7 & 3 & 2 & 120 & 55 & $\leq 327$ & 327 \\ \hline Random
& 7 & 5 & 3 & 160 & 75 & $\geq 800$ & 800 \\ \hline RS & 7 & 5 & 3 & 160 & 75
& $\leq 745$ & 745 \\ \hline
  \end{tabular} \medskip
  \caption{Parameter sets for the tests. The code $\CC$ is the
    shortening at $m-\lambda$ positions per block of the expansion of
    a parent code. The parent code is either random or a Reed--Solomon
    code, as indicated in the first column of the table. The
    penultimate column gives the bounds on the dimension of the
    twisted square code shortened at $\KLM$: a lower bound for random
    codes (Corollary~\ref{cor:dim_twisted_square_shortened_random})
    and an upper bound for SSRS codes
    (Corollary~\ref{cor:dim_sq_SSRS}).  The last column gives the
    actual dimension of the twisted square code computed using
    \emph{Sage}. For each set of parameters, at least 100 tests were
    run and the actual dimension never differed from the bounds.  We
    observe in particular that the bounds stated in
    Corollaries~\ref{cor:dim_twisted_square_shortened_random}
    and~\ref{cor:dim_sq_SSRS} are typically equalities.}
  \label{tab:param_tests}
\end{table}

In particular, these experiments show that bounds
(\ref{eq:typical_square}) and~(\ref{eq:dim_square_SSRS}) are typically
equalities. Note that this observation is not necessary to distinguish
the codes but it will be useful for the attack presented in
Section~\ref{sec:attack}.

\subsubsection{Broadening the range of the distinguisher by
  shortening}\label{subsec:broadening_by_shortening}

Similarly to the works \cite{CGGOT14,COT17}, the range of the
distinguisher can be broadened by shortening the public code.  This
can make the distinguisher work in some cases when $2k \geq n$. The
idea is to shorten some blocks of length $\lambda$ (corresponding to a
given position of the original code in $\Fqm^n$). For each shortened
block the degree $k$ is decreased by~$1$. Indeed, from
Lemma~\ref{lem:pu-sh-commute} shortening a whole block corresponds to
shortening the corresponding position of the parent code over
$\Fqm$. 

Let us investigate the condition for this to work. Let $s_0$ be the
least positive integer such that $2(k-s_0)-1 < n-s_0$, \ie
\[
s_0\eqdef 2k-n.
\]
If one shortens the public code at $s \geq s_0$ blocks, which
corresponds to $s(m-\lambda)$ positions, we can apply
Theorem~\ref{thm:distinguisher} on the shortened code. The condition
of the theorem becomes
\[
  m(2(k-s)-1)  < 
  \min \left\{m(n-s), {m(k-s)- (n-s)(m-\lambda) + 1 \choose 2} - (n-s)\left(
      {\lambda + 1 \choose 2} - m\right)\right\}.
\]

\begin{example} Consider the parameters of XGRS in the first row of
  Table~\ref{tab:parameters}. Suppose we shorten $s = 820$ blocks of
  the public key (\ie $1260$ positions of the parent GRS code). It
  corresponds to reduce to $n' = n-s = 438$ and $k' = k-s = 211$.  The
  shortened public key will have dimension $195$.  Thus, the twisted
  square of the shortened public key will typically have dimension
  $1263$ while the twisted square of an expanded subspace subcode of a
  random code would have full length, \ie $3 (n-s) = 1314$.
\end{example}

\subsubsection{Limits of the distinguisher: the ``$m/2$ barrier''.}
Suppose that $\lambda \leq \frac m 2$ and let $\CC$ be a GRS code of
dimension $k$ and $\SS$ a subspace of dimension $\lambda$ such that
the SSRS code reaches the typical dimension (see
Propositions~\ref{prop:SS_parameters} and~\ref{prop:typical_dim_SS}),
\ie $\dim_{\Fq} \CCSS = km - n(m-\lambda)$.

For this dimension to be positive, we must have
\[ k > n\left(1 - \frac{\lambda}{m}\right) > \frac n 2\cdot \]

This is incompatible with the necessary condition $2k\leq n$ (see
Remark~\ref{rem:rate}) and cannot be overcome by shortening blocks as
described in Section~\ref{subsec:broadening_by_shortening}.  Hence,
\textbf{whenever $\lambda \leq m/2$, the distinguisher is
  ineffective.}

\begin{remark} In \cite{COT14a,COT17} a distinguisher on so--called
  wild Goppa codes over quadratic extensions is established using the
  square code operation after a suitable shortening.  This corresponds
  precisely to the case $\lambda = 1$ and $m=2$ which, according to the
  previous discussion, should be out of reach of the distinguisher. The
  reason why this distinguisher is efficient for these parameters is
  precisely because the dimension of such codes significantly exceeds
  the lower bound of Proposition~\ref{prop:SS_parameters} (see
  \cite{SKHN76,COT14}).
\end{remark}

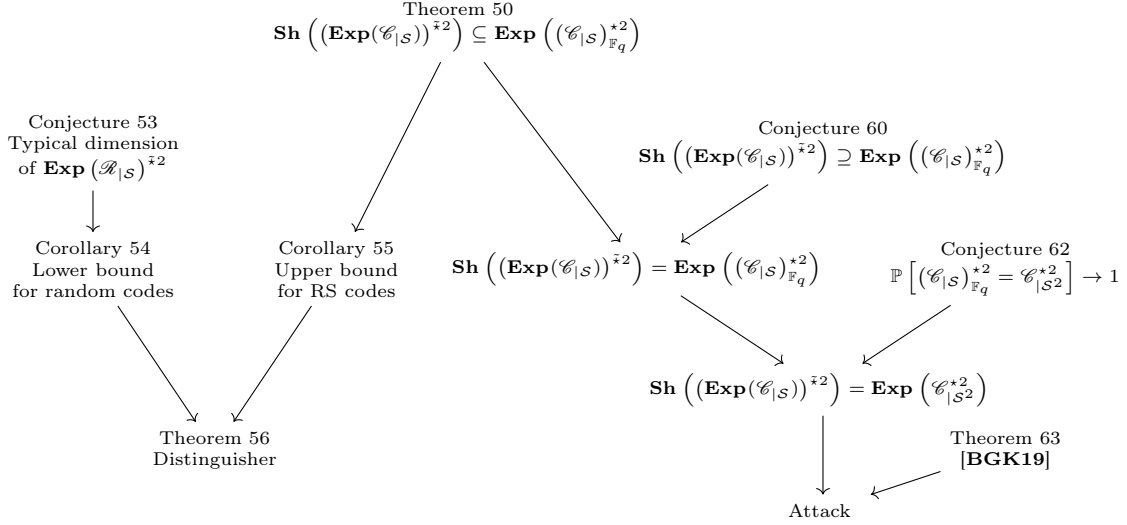
\begin{figure}[h!]
  \centering

    \begin{tikzpicture}[scale=0.8, every node/.style={minimum width=1.5cm, font=\scriptsize}]

    \node[align=center] (thm43) at (0,2)
    {
      Theorem~\ref{thm:generalise_connect}\\
      $\sh{\tsqb{\ExpCode (\SScode{\CC}{\SS})}}{} \subseteq
      \ExpCode_{}\left(\sqb{\CCSS}_{\Fq}\right)$
    };
    \node[align=center] (conj52) at (6,0) {
      Conjecture~\ref{conj:equality_generalise_connect}\\
      $\sh{\tsqb{\ExpCode(\SScode{\CC}{\SS})}}{} \supseteq
      \ExpCode_{}\left(\sqb{\CCSS}_{\Fq}\right)$
    };
    \node[align=center] (conj46) at (-6,0)
    {
      Conjecture~\ref{conj:twisted_sq}\\
      Typical dimension\\
      of $\ExpCode \tsqb{\RC_{|\SS}}$
    };
    \node[align=center] (equal1) at (3,-2)
    {
      $\sh{\tsqb{\ExpCode_{} (\SScode{\CC}{\SS})}}{} =
      \ExpCode_{}\left(\sqb{\CCSS}_{\Fq}\right)$
    };
    \node[align=center] (cor48) at (-2,-2)
    {
      Corollary~\ref{cor:dim_sq_SSRS}\\
      Upper bound\\
      for RS codes
    };
    \node[align=center] (cor47) at (-6,-2) {
      Corollary~\ref{cor:dim_twisted_square_shortened_random}\\
      Lower bound\\
      for random codes
    };
    \node[align=center] (thm49) at (-4,-5) {
      Theorem~\ref{thm:distinguisher}\\
      Distinguisher
    };
    \node[align=center] (conj54) at (9,-2)
    {
      Conjecture~\ref{conj:equality_sq_SScode}\\
      $\Prob{\sqb{\CCSS}_{\Fq} = \sq{\CC}_{|\SS^2}} \to 1$
    };
    \node[align=center] (equal2) at (6,-4)
    {
      $\sh{\tsqb{\ExpCode_{} (\SScode{\CC}{\SS})}}{} =
      \ExpCode_{}\left(\sq{\CC}_{|\SS^2}\right)$
    };
    \node[align=center] (thm57) at (9,-5)
    {
      Theorem~\ref{thm:SS_BGK}\\
      \cite{BGK19}
    };
    \node[align=center] (attack) at (6,-6)
    {
      Attack
    };

    \draw[->] (thm43) -- (equal1);
    \draw[->] (thm43) -- (cor48);
    \draw[->] (conj52) -- (equal1);
    \draw[->] (equal1) -- (equal2);
    \draw[->] (conj54) -- (equal2);
    \draw[->] (equal2) -- (attack);
    \draw[->] (thm57) -- (attack);
    \draw[->] (cor48) -- (thm49);
    \draw[->] (cor47) -- (thm49);
    \draw[->] (conj46) -- (cor47);
  \end{tikzpicture}

  
  \caption{Informal summary of the statements. Any statement is the consequence of the statements pointing to it.}
  \label{fig:recap}
\end{figure}


\section{Attacking the SSRS scheme}
\label{sec:attack}

In this section, we describe how to use these tools to attack the SSRS
scheme. For the sake of convenience, we first focus on the parameters
with $m=3$, $\lambda=2$ and then discuss the general case.

\subsection{Further conjectures for the attack}
As explained in Section~\ref{ex:num_app} our experiments show that
Inequalities (\ref{eq:typical_square}) and~(\ref{eq:dim_square_SSRS})
are typically equalities. This encourages us to state the following
two conjectures.

\begin{conjec}\label{conj:equality_generalise_connect}
  Let $\SS, \B_\SS, \B_{\SS^2}, \KLM$ be as in
  Theorem~\ref{thm:generalise_connect} and suppose that
  Equation~\eqref{eq:condition_distinguisher} is satisfied. If $\CC$
  is an $[n,k]$ GRS code, then, with high probability, the inclusion
  of Equation~\eqref{eq:twisted_sq_exp_code} is an equality, \ie
  \[
    \sh{\tsqb{\ExpCode_{\BS} (\SScode{\CC}{\SS})}}{\KLM} =
    \ExpCode_{\BSp}\left(\sqb{\CCSS}_{\Fq}\right).
  \]
\end{conjec}

In addition, the right--hand term of the last equality satisfies the
following inclusion.

\begin{lem}\label{lem:inc_sq_SScode}
  Let $\CC \subseteq \Fqm^n$ and $\SS \subseteq \Fqm$ be an
  $\Fq$--vector space. Then
  \begin{equation}\label{eq:sq_SScode}
    \sqb{\CCSS}_{\Fq} \subseteq \left(\sq{\CC}\right)_{|\SS^2}.
  \end{equation}
\end{lem}

\begin{proof}
  It suffices to observe that the result holds on $\Fq$--generators.
  Let $\av, \bv \in \CCSS$. Then, $\av \star \bv \in \sq{\CC}$.  In
  addition, for any $i \in \{0, \dots, n-1\}$, we have
  $(\av \star \bv)_i \in \SS^2$. Thus,
  $\av \star \bv \in \SScode{(\sq{\CC})}{\SS^2}$.
\end{proof}

Moreover, Inclusion~(\ref{eq:sq_SScode}) turns out to be typically an equality
in the case of GRS codes as suggested by the following conjecture.

\begin{conjec}\label{conj:equality_sq_SScode}
  Let $\SS, \B_\SS, \B_{\SS^2}, \KLM$ be as in
  Theorem~\ref{thm:generalise_connect} and suppose that
  Equation~\eqref{eq:condition_distinguisher} is satisfied. If $\CC$
  is an $[n,k]$ GRS code, then, with high probability, the inclusion
  of Equation~\eqref{eq:sq_SScode} is an equality, \ie
  \[
    \sqb{\CCSS}_{\Fq} = \left(\sq{\CC}\right)_{|\SS^2}.
  \]
\end{conjec}

\subsection{The case $m=3$ and $\lambda = 2$}
\label{subsec:m=3_lambda=2}

\subsubsection{Constructing the square code}

Let $\Cpub$ be the public code of an instance of the SSRS scheme.
This code is described by a generator matrix $\Gpub$ which is the only
data we have access to.  We know that there exist unknown spaces
$\SS_0, \dots, \SS_{n-1}$ with bases $\B_{\SS_i} = (b_{i,0},b_{i,1})$
and an RS code over $\Fqm$ such that
\[ \Cpub =
  \ExpCode_{\left(\B_{\SS_i}\right)_i}\left(\SScode{\RS{k}{\xv}}{{(\SS_i)}_i}\right).
\]

We can compute the generator matrix of the twisted square code
$\Cpub^{\tstar 2}$, which according to
Theorem~\ref{thm:twisted_sq_exp_code} is equal to
\[
  \ExpCode_{\left(\B_{\SS_i^2}\right)_i}
  \left(\sqb{\RS{k}{\xv}_{|(\SS_i)_i}}_{\Fq}\right),
\]
where $\B_{\SS_i^2} \eqdef (b_{i,0}^2, b_{i,0} b_{i,1}, b_{i,1}^2)$.
Moreover, assuming Conjecture~\ref{conj:equality_sq_SScode}, this code is
likely to be equal to
\[
  \ExpCode_{\left(\B_{\SS_i^2}\right)_i} \left(\RS{2k-1}{\xv}\right).
\]
It is important to stress that, at this stage, we do not know the
value of $\xv$ nor the $\B_{\SS_i}$ or the $\B_{\SS^2_i}$.

\subsubsection{Finding the value of $\xv$}

We now have access to a fully expanded RS code (and not a subspace
subcode) and want to use this to find the value of $\xv$. In fact, the
authors of \cite{BGK19} propose an algorithm to solve this problem, by
using a generalisation of the algorithm of Sidelnikov and Shestakov
\cite{SS92} to recover the structure of GRS codes.

\begin{thm} \cite[\S~IV.B]{BGK19}\label{thm:SS_BGK} Let
  $\xv = (x_0, \dots, x_{n-1}) \in \Fqm^n$ be a vector with distinct
  entries and $\B_0, \dots, \B_{n-1}$ be an $n$--tuple of $\Fq$--bases
  of $\Fqm$. Let
  \[ \CC = \ExpCode_{{(\B_i)}_i}(\RS{k}{\xv}).\] There exists a polynomial
  time algorithm which
  \begin{itemize}
  \item[] {\bf takes as inputs} $\CC$, three distinct elements
    $x_0', x_1', x_2' \in \Fqm$ and an $\Fq$--basis $\B_0'$ of $\Fqm$;
  \item[] {\bf and returns}    $x_3', \dots, x_{n-1}' \in \Fqm^n$ and
  $\Fq$--bases $(\B_1', \dots, \B_{n-1}')$ of $\Fqm$ such that
  \[ \CC = \ExpCode_{(\B'_0, \dots, \B'_{n-1})}(\RS{k}{(x'_0, \dots, x'_{n-1})}).\]
  \end{itemize}
\end{thm}

The principle of the algorithm is very similar to that of Sidelnikov
Shestakov. Starting from a systematic generator matrix of an expanded
Reed--Solomon code, the hidden structure of the RS code is deduced
from relations satisfied by the $m \times m$ blocks of the right hand
side of this systematic generator matrix.

\begin{remark}
  Theorem~\ref{thm:SS_BGK} asserts in particular that
  the choice of three values of the support together with one basis
  uniquely determines a pair $(\xv, {(\B_i)}_i)$ describing a code
  $\ExpCode_{{(\B_i)}_i}(\xv)$.
\end{remark}

Using this Theorem~\ref{thm:SS_BGK}, we obtain a vector $\xv'$ and
$\Fq$--bases $\B'_i$ of $\F_{q^3}$ such that
\[
  \Cpub^{\tstar 2} = \ExpCode_{{(\B'_i)}_i}(
  \RS{2k-1}{\xv'}).
\]

\begin{remark}
  Note that the value of $\xv'$ is not necessarily the same as the one
  contained in the secret key but we are looking for an equivalent
  secret key, \ie{} we only need a code description which
  allows us to decode.
\end{remark}

\subsubsection{Recovering a secret key}\label{subsec:back_to_Cpub}

Once $\xv'$ is found, there remains to find bases
$\B_{\SS_0'}, \dots, \B_{\SS_{n-1}'}$ of $2$--dimensional subspaces
$\SS_0', \dots, \SS_{n-1}' \subseteq \F_{q^3}$ such that
\[
  \Cpub = \ExpCode_{\left(\B_{\SS_i'}\right)_i}(\RS{k}{\xv'}).
\]
These bases can be obtained by solving a linear system. They are the
pairs
\[\B_{\SS_0'}=(b_0^{(0)}, b_0^{(1)}), \dots, \B_{\SS_{n-1}'} =
(b_{n-1}^{(0)}, b_{n-1}^{(1)})\]
such that
\[
  \SqueezeCode_{\left(\B_{\SS_i'}\right)_i}(\Cpub) \subseteq \RS{k}{\xv'},
\]
which can be equated as follows. Let $\Hm$ be a parity--check
matrix of $\RS{k}{\xv}$ and $\Gpub$ a generator matrix of $\Cpub$.
Let
\[
  \Bm =
  \begin{pmatrix}
    b_{0}^{(0)} &  & & (0) \\ 
    b_{0}^{(1)} &  & & \\
     & b_1^{(0)} & & \\
     & b_1^{(1)} & & \\
    & & \ddots & \\
    & & & b_{n-1}^{(0)}\\
    (0)& & & b_{n-1}^{(1)}
  \end{pmatrix}
  \in \F_{q^3}^{2n \times n}.
\]
The unknown entries of $\Bm$ are the solutions of the linear system
\begin{equation}\label{eq:final_system}
\Gpub \Bm \transpose{\Hm} = 0.
\end{equation}

There are
\begin{itemize}
\item $2n$ unknowns in $\F_{q^3}$ which yields $6n$ unknowns in $\Fq$;
\item for $(3k-n)(n-k) = O(n^2)$ equations.
\end{itemize}
Thus, the matrix $\Bm$ is very likely to be the unique solution up to
a scalar multiple. From this, we obtain a complete equivalent secret
key, which allows to decrypt any ciphertext.

\begin{remark}
  After presenting a polynomial time recovery of the structure of
  expanded GRS codes in \cite[\S~IV.B]{BGK19}, the extension to
  expanded SSRS codes is discussed \cite[\S~VI.C]{BGK19}. The
  suggested approach consists in performing a brute--force search on
  the expansion bases $\B_0,\dots, \B_{n-1}$.  But the cost of such an
  approach is exponential in $n$ and
  $\lambda$.
  Our use of the twisted square code allows to address the same problem
  in polynomial time.
\end{remark}

\subsubsection{Extending the reach of the attack by shortening blocks}
As explained in Section~\ref{subsec:broadening_by_shortening}, it may
happen that $\tsq{\Cpub} = \Fq^{3n}$, \ie the twisted square of the
public code equals the whole ambient space. In such a situation, the
distinguisher fails and so does the attack.  To overcome this issue,
it is sometimes possible to shorten a fixed number $s$ of blocks of
$\Cpub$ and apply the previous attack to this block--shortened code.

More precisely, let $\Ical \subseteq \IInt{0}{2n-1}$ be a set of
indices corresponding to a union of blocks, \ie of the form
$\Ical = \{2i_0, 2i_0+1, \dots, 2i_s, 2i_s+1\}$. We apply the previous
algorithm to the code $\sh{\Cpub}{\Ical}$ which returns
$((x'_i)_{i \notin \Ical}, (\B'_i)_{i \not \in \Ical})$ such that
\[
  \tsq{\sh{\Cpub}{\Ical}} = \ExpCode_{(\B'_i)_{i \notin
      \Ical}}((x'_i)_{i \notin \Ical}).
\]
Recall that the choice of three of the $x_i'$'s and one of the
$\B_i$'s entirely determines the other ones. Then, one can re-apply
the same process with another set of blocks $\Ical_1$ such that there
are at least $3$ that are neither in $\Ical_0$ nor in $\Ical_1$.  This
allows to deduce new values for $x_i$'s for
$i \in \Ical \setminus \Ical_1$.  And we repeat this operation until
$\xv'$ is entirely computed. Then, we proceed as in
Section~\ref{subsec:back_to_Cpub} to recover the rest of the secret
key.

\subsubsection{Application: attacking some parameters of the XGRS system}
\label{subsec:Appl_attack_XGRS}
The proposed attack efficiently breaks all parameters of Type I
proposed in \cite{KRW21} (\ie with $\lambda = 2$ and $m=3$).  Using a
{\em Sage} implementation, the calculation of $\tsq{\Cpub}$ takes a
few minutes.  Next, we obtained a full key recovery using the ``guess
and squeeze'' approach described further in
Section~\ref{sec:guess_and_squeeze} followed by a usual Sidelnikov
Shestakov attack. The overall attack runs in less than one hour for
keys corresponding to a claimed security level of 256 bits. The
previously described approach consisting in applying directly the
algorithm of \cite[\S~VI.B]{BGK19} on $\tsq{\Cpub}$ has not been
implemented but is probably even more efficient.

\subsection{The general case}
The attack presented in Section~\ref{subsec:m=3_lambda=2} generalises
straightforwardly (up to the following details) when the conditions of
Conjecture~\ref{conj:equality_generalise_connect} are met.

\begin{itemize}
\item According to Theorem~\ref{thm:generalise_connect}, we should no
  longer work with $\tsq{\Cpub}$ but with $\sh{\tsq{\Cpub}}{\KLM}$,
  where $\KLM$ is defined in Lemma~\ref{lem:twisted_product_gen}
  (\ref{eq:J0}).  Assuming
  Conjectures~\ref{conj:equality_generalise_connect}
  and~\ref{conj:equality_sq_SScode}, we deduce that this code is the
  expansion of a GRS code. Hence, the algorithm of
  \cite[\S~VI.B]{BGK19} can be applied to it.
\item The recovery of the subspaces and bases described in
  Section~\ref{subsec:back_to_Cpub} involves a matrix
  $\Bm \in \Fq^{\lambda n\times n}$ with $\lambda n$ nonzero entries,
  which will be the unknowns of the system~(\ref{eq:final_system}).
  Hence, this system has $\lambda n$ unknowns in $\Fqm$, \ie
  $\lambda mn$ unknowns in $\Fq$ for $(mk-n(m-\lambda))(n-k) = O(n^2)$
  equations. As the value of $m$ (and hence $\lambda$) remain very
  small compared to $n$, there is still in general a unique solution
  up to a scalar multiple.
\end{itemize}

\subsection{Summary of the attack}

The attack can be summarised by the following algorithms, depending on
the values of $k$ and $n$.

\begin{algorithm}[h!]
  \caption{The attack when $2k \leq n$}
  \label{alg:algo1}
  \begin{algorithmic}[1]
    \State {Compute $\sh{\tsq{\Cpub}}{\KLM}$, where $\KLM$ is
      the the union of the last ${\lambda+1 \choose 2} - m$
      positions of each block (see Lemma~\ref{lem:JLM} (\ref{eq:def_JLM}));}
    \State { Apply the algorithm of
      \cite[\S~VI.B]{BGK19} to recover a support $\xv$ of the parent
      Reed--Solomon code;}
    \State {Apply the calculations of
      Section~\ref{subsec:back_to_Cpub} to recover the bases $\B_i$.}
  \end{algorithmic}
\end{algorithm}

\begin{algorithm}[h!]
  \caption{Attack when $2k > n$}
  \label{alg:algo2}
  \begin{algorithmic}[1]
    \State {Choose a number $s$ of blocks to shorten
      satisfying condition (\ref{eq:condition_distinguisher}) so that
      the distinguisher succeeds.}
    \State {Pick a union of $s$ blocks $\Ical$ and
      \begin{enumerate}[(a)]
      \item Compute $\sh{\tsq{\sh{\Cpub}{\Ical}}}{\KLM'}$,
        where $\KLM'$ is the union of the last ${\lambda+1 \choose 2} - m$ positions
        of each block;
      \item Apply the algorithm of \cite[\S~VI.B]{BGK19} to recover a partial
        support ${(x_i)}_{i \notin \Ical}$;
      \item Repeat this process with another $\Ical$ until you got the whole
        support $\xv$.
      \end{enumerate}
    }
    \State {Apply the calculations of Section~\ref{subsec:back_to_Cpub}
      to recover the bases $\B_i$.}
  \end{algorithmic}
\end{algorithm}

Note that the support of a GRS code is defined up to three degrees of
freedom. Hence, in Algorithm~\ref{alg:algo2}, in order to consistently
reconstruct the whole support of $\xv$, one should make sure that the
set $\Ical$ always includes the first three positions and assign them
some fixed arbitrary values.

\subsection{Complexity}
For the complexity analysis and according to the parameters proposed
in \cite{KRW21}, we suppose that $m=O(1),\ \lambda = O(1)$ and $k = \Theta (n)$.

\subsubsection{Step 1, the twisted square computation}
First let us evaluate the cost of the computation of the twisted
square of the code $\Cpub \subseteq \Fq^{\lambda n}$ of dimension
$k_0 \eqdef (mk-n(m-\lambda))) $.
\begin{enumerate}
\item Starting from a $k_0 \times \lambda n$ generator matrix
  of $\Cpub$, any non ordered pair of rows provides a generator
  of the twisted square. Hence there are
  ${k_0+1 \choose 2} = O(n^2)$ generators to compute,
  each computation costing 
  $n{\lambda+1 \choose 2}$ operations.
  This is an overall cost of $O(n^3)$ operations in $\Fq$.
\item Then, deducing a row echelon generator matrix of this twisted
  square from these $O(n^2)$ generators has the cost of the
  computation of the row echelon form of a $O(n^2) \times O(n)$
  matrix, which requires $O(n^{\omega + 1})$ operations in $\Fq$ (see
  \cite[Th\'eor\`eme~8.6]{BCGLLSS17}), where $\omega \leq 3$ is the
  complexity exponent of operations of linear algebra.
\end{enumerate}

Thus, the overall cost of the computation of this twisted square code
is $O(n^{\omega + 1})$.  In addition, in the situation where
$2k-1 > n$, we need to iterate the calculation on a constant number of
shortenings of the public code, which has no influence on the
complexity exponent.

\if\longversion1
\begin{remark}
  Similarly to the discussions \cite[\S~VI.D]{COT17} and
  \cite[\S~VI.B.4]{CMP17}, it is possible to randomly generate $O(n)$
  generators of the twisted square code and perform the echelon form
  on this subset of generators. This provides the whole twisted square
  code with a high probability, reducing the cost of the calculation
  to $O(n^{\omega})$ operations in $\Fq$.
\end{remark}
\fi

\subsubsection{Step 2, recovering $\xv$}
The second step of the attack, \ie performing the algorithm of
\cite[\S~VI.B]{BGK19} to recover $\xv$ is not that expensive. A quick
analysis of this algorithm allows to observe that the most time
consuming step is the calculation of the systematic form of the
generator matrix, which has actually been performed in the previous
step. Therefore, this second step can be neglected in the complexity
analysis.

\subsubsection{Step 3, recovering the bases}
Finally, the last step of the attack, consisting in recovering the
bases $\B_i$, boils down to the resolution of a linear system of $O(n^2)$
equations and $O(n)$ unknowns, which costs $O(n^{\omega + 1})$ operations.

\medskip

\noindent {\bf Summary.} The overall cost of the attack is of
$O(n^{\omega + 1})$ operations in $\Fq$.

\subsection{Recovering the bases for arbitrary expanded codes:
  \emph{guess and squeeze}}\label{sec:guess_and_squeeze}

To conclude this section, we present an alternative approach to detect
the hidden structure of expanded codes and recover the expansion
bases. As explained in Section~\ref{subsec:Appl_attack_XGRS}, this is
the approach we implemented. The advantage of this approach is that it
can be applied to expansions of codes which are not RS
codes. Therefore it could be an interesting tool for other
cryptanalyses.

Given a code $\CC \subseteq \Fqm^n$ and bases $\B_0, \dots, \B_{n-1}$
of $\Fqm$, suppose you only know a generator matrix of
\[
  \Cexp \eqdef \ExpCode_{(\B_i)}(\CC).
\]
The objective is to guess the $\B_i$'s
iteratively instead of brute forcing any $n$--tuple of bases, which
would be prohibitive.

\medskip

\begin{enumerate}[{\bf Step 1.}]
\item Shorten $\Cexp$ at $k-1$ blocks (which corresponds to $m(k-1)$
  positions). This yields a code whose dimension most of the times
  equals $m$. According to Lemma~\ref{lem:pu-sh-commute}, this is the
  expansion of a code of dimension $1$ obtained by shortening $\CC$ at
  $k-1$ positions.
\item Puncture this shortened code in order to keep only two blocks.
  We get a $[2m,m]$ code which we call $\Cexptiny \subseteq \Fq^{2m}$.
  This code is the expansion of a $[2, 1]$ code called
  $\Ctiny \subseteq \Fqm^2$ obtained from $\CC$ by shortening $k-1$
  positions and puncturing the remaining code in order to keep only $2$
  positions.
\item Now, for any pair of bases $(\B_0, \B_1)$ of $\Fqm$, compute
  \[
    \SqueezeCode_{(\B_0, \B_1)}(\Cexptiny).
  \]
  The point is that, for a wrong choice of bases, we get a generator
  matrix with $m$ rows and $2$ columns which is very likely to be full
  rank.  Hence a wrong choice provides the trivial code $\Fqm^2$. On
  the other hand, a good choice of bases provides the code $\Ctiny$
  which has dimension $1$.  This property allows to guess the bases.
\end{enumerate}

Actually, according to Lemma~\ref{lem:Fqm-scalar-multiplication}, if
one guesses the bases $a_0\B_0, a_1 \B_1$ for some
$a_0, a_1 \in \Fqm^\times$, the squeezing will provide
$\Ctiny \star (a_0, a_1)$ which also has dimension $1$. Therefore, it
is possible to first guess the bases up to a scalar multiple in
$\Fqm^\times$.  Therefore, the cost of computing these two bases is in
$O(q^{2m(m-1)})$ operations.

Once the first two bases are known, one
can restart the process by with another pair of blocks involving one
of the two blocks for which the basis is already known, which requires
$O(q^{m(m-1)})$ operations. This yields an overall complexity of
$O(q^{2m(m-1)} + nq^{m(m-1)})$ operations in $\Fq$ for this {\em guess
  and squeeze} algorithm.

\begin{remark}
  Note that in the attack of XGRS scheme, the bases to guess are known
  to be of the form $(1, \gamma, \gamma^2, \dots, \gamma^{m-1})$ for
  some generator $\gamma \in \Fqm$. This additional information
  significantly improves this search and reduce the cost
  of the calculation of the $n$ bases to $O(q^{2m} + n q^m)$
  operations.
\end{remark}

\begin{remark}
  Proceeding this way only allows to get back the code
  $\CC \star \av \subseteq \Fqm^n$ for an unknown vector
  $\av \in {(\Fqm^\times)}^n$. However, this is an important first
  step.  For instance, if $\CC$ was a Reed--Solomon, we obtain a
  generalised Reed--Solomon code whose structure is computable using
  Sidelnikov and Shestakov attack. It is then possible to decode.
\end{remark}


\section{Conclusion}

We presented a polynomial time distinguisher on subspace subcodes of
Reed--Solomon codes relying on a new operation called the {\em twisted
  square product}.  We are hence able to distinguish SSRS codes from
random ones as soon as the dimension $\lambda$ of the subspaces
exceeds $\frac m 2$.  From this distinguisher, we derived an attack
breaking in particular the parameter set $\lambda = 2$ and $m=3$ of
the XGRS system \cite{KRW21}.

These results contribute to better understand the security of the
McEliece encryption scheme instantiated with algebraic codes. On the
one hand, we have generalised Reed--Solomon codes, which are known to
be insecure since the early 90's. On the other hand, alternant codes
seem to resist to any attack except some Goppa codes with a very low
extension degree \cite{COT17, FPP14}. The present work provides an
analysis of a family of codes including these two cases as the two
extremities of a spectrum. Concerning the subspace subcodes lying in
between, we show an inherent weakness of SSRS codes when
$\lambda > m/2$ (See Figure~\ref{fig:fleche}, page
\pageref{fig:fleche}).  The case $\lambda = m/2$ is in general out of
reach of our distinguisher, but remains border line as testified by
some attacks on the cases $\lambda = 1, m = 2$ in the literature
\cite{COT17, FPP14}.

A question which remains open is the actual security of the cases
$1 < \lambda < m/2$ which are out of reach of the twisted square code
distinguisher. These codes, which include alternant codes, deserve to
have a careful security analysis in the near future. Indeed, if they
turn out to be resistant to any attack, they could provide an
alternative to {\em Classic McEliece} \cite{BCLMNPPSSSW19} with
shorter key sizes. On the other hand, if some of these codes turned
out to be insecure, this may impact the
security of {\em Classic McEliece} which is a crucial question in the
near future.


\newcommand{\etalchar}[1]{$^{#1}$}

\end{document}